\documentclass{article} 

\usepackage{graphicx}
\usepackage{mathtools}
\mathtoolsset{showonlyrefs}
\usepackage{amsmath}
\usepackage{amsthm}
\usepackage{algorithm}
\usepackage{a4wide}
\usepackage{commath}
\usepackage{booktabs}

\newcommand{\E}{\mathbf{E}}
\newcommand{\cov}{\mathbf{cov}}
\newcommand{\N}{\mathcal{N}}
\newcommand{\Q}{\mathcal{Q}}
\newcommand{\T}{\mathrm{T}}
\newcommand{\p}{\mathrm{p}}
\newcommand{\q}{\mathrm{q}}
\DeclareMathOperator{\trace}{\mathrm{Tr}\!}
\newcommand{\mc}{\text{\textsc{mc}}}
\newcommand{\vb}{\text{\textsc{vb}}}

\theoremstyle{plain}
\newtheorem{thm}{Theorem}
\newtheorem{cor}{Corollary}
\newtheorem{lem}{Lemma}

\newtheorem{prob}{Problem}
\theoremstyle{remark}
\newtheorem{rem}{Remark}

\allowdisplaybreaks%

\begin{document}

\title{Modeling and identification of uncertain-input systems}

\author{Riccardo Sven Risuleo\\\small{ACCESS Linnaeus Center, KTH Royal Institute of Technology,
    Sweden}\\\texttt{risuleo@kth.se}\and Giulio Bottegal\\\small{Department of Electrical Engineering, Eindhoven University of
  Technology, Eindhoven, The Netherlands}\\\texttt{g.bottegal@tue.nl}\and H\r{a}kan
  Hjalmarsson\\\small{ACCESS Linnaeus Center, KTH Royal Institute of Technology,
  Sweden}\\\texttt{hjalmars@kth.se}}             

\date{\today}

\maketitle{}

\begin{abstract} 
  In this work, we present a new class of models, called uncertain-input
  models, that allows us to treat system-identification problems in which a
  linear system is subject to a partially unknown input signal. To encode prior information about the input or the linear system, we
  use Gaussian-process models.  We estimate the model from data using the
  empirical Bayes approach: the input and the impulse responses of the
  linear system are estimated using the posterior means of the Gaussian-process
  models given the data, and the hyperparameters that characterize the
  Gaussian-process models are estimated from the
  marginal likelihood of the data. We propose an iterative algorithm to find
  the hyperparameters that relies on the EM method and results in simple update
  steps. In the most general formulation, neither the marginal likelihood nor
  the posterior distribution of the unknowns is tractable. Therefore, we
  propose two approximation approaches, one based on Markov-chain Monte Carlo
  techniques and one based on variational Bayes approximation. We also show special model
  structures for which the distributions are treatable exactly. Through
  numerical simulations, we study the application of the uncertain-input model
  to the identification of Hammerstein systems and cascaded linear systems.
  As part of the contribution of the paper, we show that this model structure
  encompasses many classical problems in system identification such as
  classical PEM, Hammerstein models, errors-in-variables problems, blind system
  identification, and cascaded linear systems. This allows us to build
  a systematic procedure to apply the algorithms proposed in this work to a
  wide class of classical problems.
\end{abstract}

\section{Introduction}
In most system identification problems, the input signal---that is, the
independent variable---is perfectly known~\cite{ljung1999system}. Often, the
input signal is the result of an identification experiment, where a signal with
certain characteristics is designed and applied to the system to measure its
response. However, in some applications, the hypothesis that the input signal is
known may be too restrictive. In this work, we propose a new model structure
that accounts for partial knowledge about the input signal and we show how many
classical system identification problems can be seen as problems of identifying
instances of this model structure.

The proposed model structure, which we call \emph{uncertain-input model}, is composed of a linear
time-invariant dynamical system (the \emph{linear system}) and of a signal
 of which partial information is available (the \emph{unknown input}). In the
next section, we characterize formally the unknown input; before that, we give
some examples of classical models that can be seen as uncertain-input models.
The \emph{Hammerstein model} is a cascade composition of a static nonlinear
function followed by a linear time-invariant dynamical
system~\cite{bai1998optimal,giri2010block,risuleo2015kernel}. In the Hammerstein
model, the (perfectly known) input signal passes through the unknown static
nonlinear function. After the nonlinear transformation, the signal which
is fed to the linear system, is completely unknown.
However, some characteristics of the signal may be known; for instance, we may
known that the nonlinear function is smooth or we may have a set of candidate
basis functions to choose among. Another instance of a
model where the input is not perfectly known is the \emph{errors-in-variables}
model~\cite{soederstroem2010system}. In the errors-in-variables formulation, the
input in known up to noisy measurements. The noise in the input introduces many difficulties and
special techniques have been developed to deal with
it~\cite{soederstroem2003why,soederstroem2007errors}. Closely related to
errors-in-variables models, \emph{blind system identification} methods are
used when the input signal is completely unknown~\cite{abedmeraim1997blind}.
These are particularly useful in telecommunications, image reconstruction, and
biomedical
applications~\cite{nakajima1993blind,moulines1995subspace,mccombie2005laguerre}.
Blind system identification problems are generally ill posed, and certain
assumptions on the input signal are needed to recover a
solution~\cite{ahmed2014blind}. Similar to blind problems are the problems of
\emph{system identification with missing data}. In these cases, the missing
data are estimated, together with a description of the system, by making
hypotheses on the mechanism that generated the missing
data~\cite{wallin2014maximum,markovsky2013structured,pillonetto2009bayesian,risuleo2016kernel,linder2017identification}.

In all the applications we have outlined, we can identify the common thread of
a linear system fed by a signal about which we have limited \emph{prior
information}. This leads us naturally to consider a Bayesian
framework where we can use prior distributions to encode beliefs about the
unknown quantities~\cite[Section~2.4]{bernardo2000bayesian}. Within the vast
framework of Bayesian methods, we concentrate on Gaussian
processes~\cite{rasmussen2006gaussian}. These enable us to
compute many quantities in closed form and to reason about identification in
terms of a limited number of sufficient statistics. For these reasons,
Gaussian-process modeling has become a popular approach in system
identification~\cite{pillonetto2011new,frigola2014identification,svensson2017flexible}. Although Gaussian processes
are typically analytically convenient, the structure of the uncertain-input problem leads
to an intractable inference problem: even though we model the system and the
input as Gaussian processes, the output of the system depends on their
convolution and therefore does not admit a Gaussian description. To perform the
inference---that is find the posterior distribution of the unknowns given the
observations---we need approximation methods. We propose two different
approximation methods for the posterior distribution of the unknowns: one Markov
Chain Monte Carlo~(MCMC, see~\cite{gilks1996markov}) method and one variational
 approximation~\cite{beal2003variational} method. In the MCMC
 method, we use the Gibbs sampler~\cite{geman1984stochastic} to draw
 particles from the posterior distribution and we approximate expectations as
 averages computed with the particles. In the variational method, we find the
 factorized distribution that best approximates the posterior distribution in
 Kullback-Leibler distance.

To give flexibility to the model, we allow the Gaussian priors to depend on
certain parameters (called \emph{hyperparameters}) that need to be estimated
from data together with the measurement noise variances. To estimate these
parameters, we use the \emph{empirical Bayes method} which requires maximizing
the marginal distribution of the data (sometimes called \emph{evidence},
see~\cite{maritz1989empirical}). To
this end, we propose an iterative algorithm based on the
Expectation-Maximization
(EM) method~\cite{dempster1977maximum}. The EM method alternates between the
computation of the expected value of the joint likelihood of the data, of the
unknown system, and of the input (E-step), and the maximization of this expected
value with respect to the unknown parameters (M-step). We show that the E-step
can be computed using the same approximations of the posterior distributions
that are used in the inference and that the M-step consists in a series of
simple and independent optimization problems that can be solved easily.

As mentioned above, the uncertain-input model encompasses several classical
model structures that have been object of research in the system-identification
community for decades. Two important contributions of this work are as follows.
\begin{enumerate}
  \item We unify the problems of identifying systems that are usually regarded
    as belonging to different model classes into a single identification
    framework.
  \item We formalize a method to apply the new tools of Gaussian processes and
    Bayesian inference to classical system identification problems.
\end{enumerate}

To support the validity of the proposed methods, we present identification
experiments on synthetic datasets of cascaded linear systems and of
Hammerstein systems.

\subsection{Notation}
The notation $[A]{}_{i,j}$ indicates the element
of matrix $A$ in position $i,j$ (single subscripts are used for vectors).
``$\T_{N\times n}(v)$'' denotes the $N$ by
$n$ lower-triangular Toeplitz matrix of the $m$ dimensional vector $v$:
\begin{equation}
\sbr{\T_{N\times n}(v)}_{i,j} = \begin{cases} v_{i-j+1} & 0\leq i-j+1\leq m\\0
                                                        & \text{otherwise}\end{cases}
\end{equation}
If $v$ is a vector, then $V$ is the $N$ by $N$ Toeplitz matrix whose
elements are given by $v$. The notation ``$\|a\|^2_{M}$'' is shorthand for $a^T Ma$. The
notation ``$\N(\alpha,\Sigma)$'' indicates the Gaussian distribution with mean
vector $\alpha$ and covariance matrix $\Sigma$. The notation
``$\mathcal{GP}(\mu,\Sigma)$'' indicate a Gaussian process with mean function
$\mu$ and covariance function $\Sigma$. Random variables and their realizations
have the same symbol. The notation ``$x;\theta$'' indicates that the random variable
$x$ depends on the parameter $\theta$. If $x$ is a random variable, $\p(x)$
denotes its density. The symbol ``$\cong$'' indicates equality up to an
additive constant and ``$\delta$'' is the Dirac density.

\section{Uncertain-input systems}\label{sec:ui_systems}
In this work, we propose a new model structure called the \emph{uncertain-input
model}. Consider the block scheme in Figure~\ref{fig:ui_block_scheme}. Many system
identification tasks can be formulated as the identification of a linear system $S$,
subject to an input sequence $\{w_t\}$. In this work, we consider problems in
which we have partial information about the input sequence, and this partial
information depends on the specific problem at hand.

\begin{figure}[htb]
  \centering
  \includegraphics{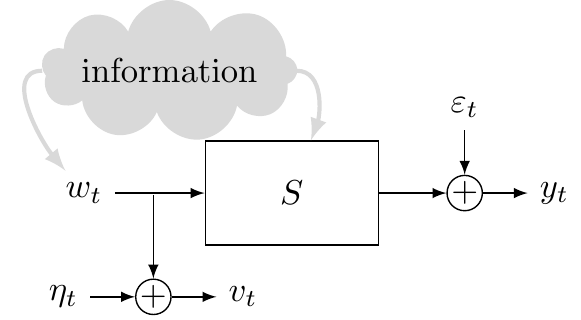}
  \caption{A block scheme of the general uncertain input system.}\label{fig:ui_block_scheme}
\end{figure}

We assume that the linear system $S$ is time invariant, stable, and causal. Therefore,
it is uniquely described by the sequence $\{g_t\}$ of its impulse response
samples, and the output of the system generated by an input $\{w_t\}$ can be
represented as the discrete convolution of the system impulse response with the
input signal---that is, at time $t$, the measurements of the output can be
written as the noise-corrupted discrete convolution
\begin{equation}\label{eq:model_y}
  y_t = \del{w\ast g}_t + \varepsilon_t,
\end{equation}
where $\{\varepsilon_t\}$ is a stochastic process that describes additive
measurement noise, and where ``$\ast$'' denotes the discrete time convolution
\begin{equation}\label{eq:convolution}
  \del{w\ast g}_t = \sum_{k=1}^\infty g_{k}w_{t-k}.
\end{equation}

In the uncertain-input model, we consider that the input signal is measured with
additive white noise described by a stochastic processes $\{\eta_t\}$. This
assumption allows us to write, for the input measurements, the model
\begin{equation}\label{eq:model_v}
  v_t = w_t + \eta_t.
\end{equation}

We assume that the noise processes $\{\eta_t\}$ and $\{\varepsilon_t\}$ are independent
Gaussian white-noise processes. This means that every noise sample has a
Gaussian distribution,
\begin{equation}\label{eq:model_noise}
    \eta_t \sim \mathcal{N}(0,\sigma_v^2),\quad
    \varepsilon_t \sim \mathcal{N}(0,\sigma_y^2),
\end{equation}
and that $\varepsilon_t$ is independent of $\varepsilon_s$, for
$s\neq t$, and of $\eta_s$ for any $s$. To allow for models where some
observations are missing, we assume infinite variance for those noise
components that correspond to the missing samples.

To encode the prior information we have about the input signal and about the linear
system, we use Gaussian process models. We model the unknown input signal
and the impulse response of the linear system as a realization of a joint Gaussian
processes with suitable mean and covariance functions,
\begin{equation}\label{eq:gp_joint}
\sbr[4]{\begin{matrix}w\\g\end{matrix}}\!\!\sim \!\mathcal{GP} \!\del{\!\!
    \sbr[4]{\begin{matrix}\mu_w(\cdot;\theta)\\
    \mu_g(\cdot;\rho)\end{matrix}}\!\!,\!
    \sbr[4]{\begin{matrix}K_w(\cdot;\theta) &
        {K_{gw}(\cdot,\cdot;\rho,\theta)}^T\\
      K_{gw}(\cdot,\cdot;\rho,\theta) &
      K_{g}(\cdot,\cdot;\rho)
  \end{matrix}}
  }\!.
\end{equation}
The mean functions of the Gaussian processes, $\mu_g(\cdot\,;\theta)$ and
$\mu_w(\cdot\,;\rho)$, may depend on the parameter vectors $\theta$ and
$\rho$, called \emph{hyperparameter vectors}, which can be used to shape the prior
information to the specific application. The same goes for the covariance
functions $K_w(\,\cdot\,,\,\cdot\,;\theta)$,
$K_g(\,\cdot\,,\,\cdot\,;\rho)$, and $K_{gw}(\,\cdot\,,\,\cdot\,;\rho,\theta)$
which may depend on (possibly different) hyperparameters.

For notational convenience, we present the explicit computations in the case of
independent Gaussian process models for $g$ and $w$---that is, we consider the
case where
\begin{equation}\label{eq:gp_models}
  \begin{split}
    w &\sim
    \mathcal{GP}\big(\mu_w(\,\cdot\,\,;\theta),K_w(\,\cdot\,,\,\cdot\,;\theta)\big),\\
    g &\sim \mathcal{GP}\big(\mu_g(\,\cdot\,;\rho),K_g(\,\cdot\,,\,\cdot\,;\rho)\big),
  \end{split}
\end{equation}
and the cross-covariance of processes
is zero. However, all results we show hold also in the more general case.

We assume that we have collected $N$ measurements of the processes $\{v_t\}$
and $\{y_t\}$ and, for sake of simplicity, we also assume that $w_t=0$ for
$t<0$ (see~\cite{risuleo2015estimation} for a way to extend the proposed
framework to unknown initial conditions).
From~\eqref{eq:model_y}, we see that the output measurements only depend on the
values of the impulse response at the discrete time instants $t=1,2,\ldots,N$; therefore,
we can consider the joint distribution of the samples $g_t$ for $t=1,2,\ldots,N$.
From the Gaussian process model~\eqref{eq:gp_models}, we have that, if we
collect the samples of $\{g_t\}$ into an $N$-dimensional column vector $g$,
this vector has a joint Gaussian distribution given by
\begin{equation}\label{eq:model_g}
    g \sim \N\big(\mu_g(\rho),K_g(\rho)\big),
\end{equation}
where we have defined the mean vector and the
covariance matrix induced by~\eqref{eq:gp_models} as
\begin{equation}
  {\big[\mu_g(\rho)\big]}_j := \mu_g(j\,;\rho), \quad {\big[K_g(\rho)\big]}_{i,j} :=
  K_g(i,j\,;\rho).
\end{equation}

From~\eqref{eq:model_v} and~\eqref{eq:model_y}, we have that the $N$
measurements of the input and output only depend on the samples $w_t$ for
$t=1,\ldots,N$; therefore, we can consider the joint distribution of these
samples, collected in an $N$-dimensional vector $w$.
This distribution is Gaussian, and it is given by
\begin{equation}\label{eq:model_w}
    w \sim \N\big(\mu_w(\theta),K_w(\theta)\big),
\end{equation}
where we have defined the mean vector and the
covariance matrix induced by~\eqref{eq:gp_models} as
\begin{equation}
  {\big[\mu_w(\theta)\big]}_j := \mu_w(j\,;\theta), \quad {\big[K_w(\theta)\big]}_{i,j} :=
  K_w(i,j\,;\theta).
\end{equation}

Assembling the models for the different components, given
by~\eqref{eq:model_y},~\eqref{eq:model_v},~\eqref{eq:model_noise},~\eqref{eq:model_g},
and~\eqref{eq:model_w}, we arrive at the following definition
of the uncertain-input model:
\begin{equation}\label{eq:ui_model}
  \left\lbrace
  \begin{aligned}
    y &= Wg + \varepsilon\\
    v &= w + \eta\\
    g & \sim \N\big(\mu_g(\rho),K_g(\rho)\big)\\
    w & \sim \N\big(\mu_w(\theta),K_w(\theta)\big)\\
    \varepsilon & \sim \N\big(0,\sigma_y^2 I_N\big)\\
    \eta & \sim \N\big(0,\sigma_v^2 I_N\big)\\
    g,\,&w,\,\varepsilon,\,\eta\text{ mutually independent}
  \end{aligned}
  \right.
\end{equation}
where we have collected the output measurements $\{y_t\}$ in a vector $y$ and
where $\varepsilon$ and $\eta$ are the vectors of the first $N$ input and output
noise samples.  The matrix
$W$ is the $N\times N$ Toeplitz matrix of the input,
$W := \T_{N\times N}(w)$,
 which represents the discrete-time convolution~\eqref{eq:convolution} as the
product $Wg$. If we define the $N\times N$ Toeplitz matrix of the impulse
response samples, $G := \T_{N\times N}(g)$, then we have the property
\begin{equation}\label{eq:toeplitz_property}
  Wg = Gw.
\end{equation}
In the next section, we give examples of some classical system identification
problems that can be cast as uncertain-input identification problems.

\section{Examples of uncertain-input models}
The uncertain-input framework is a generalization of many classical
system-identification problems. All these classical problems can be analyzed
using the tools of uncertain-input models; furthermore, under the right
conditions, the identification approach that we propose for uncertain-input
models reduces to classical system-identification approaches.
\subsection{Linear predictor model}
Consider the \emph{output-error} transfer-function model~\cite{ljung1999system},
\begin{equation}
  y_t = \frac{B(q;\rho)}{F(q;\rho)} u_t + \varepsilon_t,
\end{equation}
where $B(q;\rho)$ and $F(q;\rho)$ are polynomials in the one-step shift operator
$q$ and $\varepsilon_t$ is Gaussian white noise. If we consider the parametric
predictor of the output-error model, we can write it as
\begin{equation}
  \hat y_{t|t-1} = \frac{B(q;\rho)}{F(q;\rho)} u_t = \del{g(\rho)\ast u}_t,
\end{equation}
where $g_t(\rho)$ is the impulse response of the predictor transfer function.
We can see this model as a degenerate uncertain-input model with
\begin{equation}
  \begin{aligned}
    {[\mu_g(\rho)]}_i &= g_i(\rho), & {[K_g(\rho)]}_{i,j} &= 0,\\
    {[\mu_w(\theta)]}_i &= u_i, & {[K_w(\theta)]}_{i,j} &= 0.
  \end{aligned}
\end{equation}
We can also incorporate the framework of Bayesian identification of finite
impulse-response models with first order stable-spline kernels (for a survey,
see~\cite{pillonetto2014kernel}) with the choice
\begin{equation}\label{eq:stable-spline}
  \begin{aligned}
    {[\mu_g(\rho)]}_i &= 0, & {[K_g(\rho)]}_{i,j} &= \rho_1\,\rho_2^{\max(i,j)},\\
    {[\mu_w(\theta)]}_i &= u_i, & {[K_w(\theta)]}_{i,j} &= 0,
  \end{aligned}
\end{equation}
where $\rho_1\geq 0$ is a scaling parameter and $\rho_2\in \intcc{0,1}$
regulates the decay rate of $g$ (see,~\cite{pillonetto2010new}).
Note that, in this formulation, any kernel can be used to model $g$ (see, for
instance,~\cite{chen2014constructive,dinuzzo2015kernels}).

\subsection{Errors-in-variables system identification}
Errors-in-variables models are often described by the set of
equations~\cite{soederstroem2007errors},
\begin{align}
  y_t &= \del{g \ast w}_t + \eta_t,\\
    v_t &= w_t + \varepsilon_t.
\end{align}
It is clear that this type of models naturally fit into the
uncertain-input framework of~\eqref{eq:ui_model}.
In particular, we can consider the classical errors-in-variables
problem of identifying a parametric model of $S$ when $w_t$ is the realization of
a stationary stochastic signal with a rational
spectrum~\cite{castaldi1996identification}. In this case, we can
write $\{w_t\}$ as the filtered white noise process
\begin{equation}
  w_t = \frac{C(q;\theta)}{D(q;\theta)}e_t, 
\end{equation}
where $e_t$ is unitary variance Gaussian white noise, 
and $C(q;\theta)$ and $D(q;\theta)$ are complex polynomials in the one-step
shift operator $q$.

From this expression, we see that $w$ is a Gaussian process with
zero mean and covariance matrix $\Sigma_w(\rho)$ that depends on the
parameterization of the input filter. Using a parametric model for the
system, we obtain an uncertain-input system with
\begin{equation}
  \begin{aligned}
    {[\mu_g(\rho)]}_i &= g_i(\rho), & {[K_g(\rho)]}_{i,j} &= 0,\\
    {[\mu_w(\rho)]}_i &= 0, & {[K_w(\rho)]}_{i,j} &= \Sigma_w(\rho).
  \end{aligned}
\end{equation}
Alternatively, we could estimate all samples of the input signal
with the choice ${[\mu_w(\rho)]}_i = \theta_i$ and
$[K_w(\rho)]=0$, even though this may lead to nonidentifiability of the
model~\cite{soederstroem2003why,zhang2015errors,risuleo2016kernel}.

\subsection{Blind system identification}
Blind system identification can also be cast as the problem of identifying an
uncertain-input model by setting the input noise variance to $\sigma_v^2 =
\infty$ (this indicates that no input measurements are available). In this
case, different parameterizations of the input
lead to different models for the input process. For instance, we can consider
the parameterization of the input as a switching signal with known switching
instants $T_0<T_1<\cdots<T_p$; in this case we can choose
${[\mu_w(\rho)]}_i = h_i^t\theta$,
where $h_i$ is a selection vector that is nonzero in the $i$th interval:
\begin{equation}
  \begin{cases}{}
    {[h_i]}_j = 1 & \text{if}\;\;T_{i-1}< t \leq T_i,\\
    {[h_i]}_j = 0 & \text{otherwise}.
  \end{cases}
\end{equation}
Models similar to this one were used, for instance, in~\cite{ohlsson2014blind}
and~\cite{bottegal2015blind}.

\subsection{Cascaded system identification}\label{sec:cascaded_systems}
In cascaded linear systems, the output of one linear system is used as the input
to a second linear system (see Figure~\ref{fig:cascade}).
\begin{figure}[htb]
  \centering
  \includegraphics{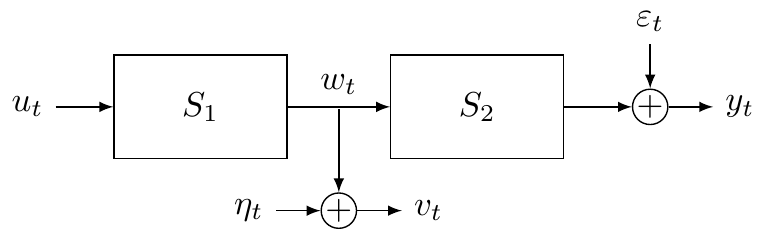}
  \caption{Cascaded linear systems.}\label{fig:cascade}
\end{figure}

For sake of argument, we consider nonparametric models for both linear systems
(the reasoning also holds for parametric models):
\begin{equation}
  g_1 \sim \N\big(0,K_1(\theta)\big),\qquad g_2 \sim \N\big(0,K_2(\rho)\big).
\end{equation}
Because $g_1$ is a Gaussian vector, the intermediate variable $w$ is also a
Gaussian vector, with zero mean and covariance matrix given by
\begin{equation}\label{eq:cascaded_input}
  K_w(\theta) = UK_1(\theta) U^T,
\end{equation}
where $U:=\T_{N\times N}(u)$ is the Toeplitz matrix of the input signal $u_t$. Therefore, we
can model the linear cascade as an uncertain-input model with input modeled as a
zero-mean process with covariance matrix given by~\eqref{eq:cascaded_input}
where, for instance, we use the first-order stable spline kernel introduced
in~\eqref{eq:stable-spline}. The same choice of kernel can be made for $K_2(\rho)$.

\subsection{Hammerstein model identification}\label{sec:hammersteins}
The Hammerstein model is a cascade of a static nonlinear function followed by a
linear dynamical system (see Figure~\ref{fig:hammerstein}).
\begin{figure}[htb]
  \centering
\includegraphics{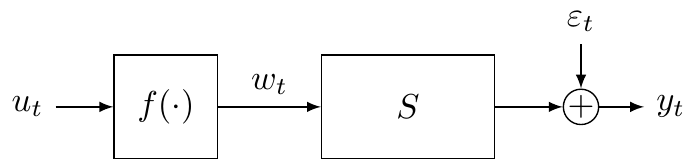}
  \caption{The Hammerstein model.}\label{fig:hammerstein}
\end{figure}

In the Hammerstein model, the intermediate variable $w_t$ is not observed (which,
symbolically, corresponds to an infinite $\sigma_v^2$). If we consider models
for the input block  that are  combinations of known basis
functions~\cite{bai1998optimal}, according to
\begin{equation}
  w_t = \sum_{j=1}^p \theta_j\, \varphi_j(u_t)
\end{equation}
we can collect the values of the unknown input $w$ and the parameters in a vector such that
\begin{equation}
  w = \Phi \theta,\qquad {[\Phi]}_{i,j} = \varphi_j(u_i).
\end{equation}
This can be modeled as an uncertain-input model with
$\mu_w(\theta) = \Phi\theta$ and $K_w(\theta) = 0$.

The uncertain-input framework also encompasses nonparametric models for the
input nonlinearity. For instance, we can model the Hammerstein cascade as the
uncertain-input model with the Gaussian radial-basis-function kernel as input model:
\begin{equation}\label{eq:rbf}
  {[K_w(\theta)]}_{i,j} = \theta_1
  \exp\left\{-\frac{1}{\theta_2}{(u_i-u_j)}^2\right\}.
\end{equation}
As for the linear system, we can use either parametric or nonparametirc
modeling approaches (see~\cite{bai1998optimal,risuleo2015new}).

\section{Estimation of uncertain-input models}\label{sec:estimation}
As discussed in Section~\ref{sec:ui_systems}, we suppose that we have collected
$N$ samples of the output $y_t$ and, possibly, $N$ samples of the noisy input
signal $v_t$ (in some applications, such as Hammerstein models and blind system
identification, these samples are not available). Whenever present, we assume
that the external input $u_t$ is completely known. We consider the following
identification problem.

\begin{prob}\label{prob:sysid}
  Given the $N$-dimensional vectors of measurements $y$ and $v$,  generated
  according to~\eqref{eq:ui_model}, estimate the impulse response $g$, the
  unknown input $w$, and the hyperparameters
  $\tau = \{\rho,\theta,\sigma_y^2,\sigma_v^2\}$.
\end{prob}
Because we are using the Gaussian process model~\eqref{eq:gp_models}, we have
natural candidates for the estimates of $g$ and $w$. Interpreting~\eqref{eq:model_g} and~\eqref{eq:model_w} as prior distributions of the
unknowns, we know that the best estimates given the data (in the minimum
mean-square error sense) are the conditional expectations
\begin{equation}
  g^\star = \E\big[ g\big| y,v\big],\quad w^\star = \E\big[w\big|y,v\big].
\end{equation}
However, these conditional expectations depend on the value of the
hyperparameter vector $\tau$. Because this value is not available, we follow an
empirical Bayes approach~\cite{maritz1989empirical} and we approximate
the true conditional expectations---that correspond to the true values of the
hyperparameters $\tau$---with the conditional expectations
\begin{equation}\label{eq:solution_posteriormean}
    \hat g := \int g\, \p(g|y,v\,;\hat\tau) \dif g,\qquad
    \hat w := \int w\, \p(w|y,v\,;\hat\tau) \dif w,
\end{equation}
where we are using estimated values $\hat \tau$ of the hyperparameters. In the
empirical Bayes approach, the estimates of the hyperparameters are chosen
by maximizing the marginal likelihood of the data,
\begin{equation}\label{eq:solution_marginal}
  \hat \tau := \arg \max_{\tau} \log \p(y,v\,;\tau),
\end{equation}
where $\p(y,v\,;\tau)$ is the marginal distribution of the measurements according
to the model in~\eqref{eq:ui_model}.

Solving~\eqref{eq:solution_marginal} yields the marginal likelihood estimate of
the hyperparameters that can be used to find the empirical Bayesian estimates of
$g$ and $w$ in~\eqref{eq:solution_posteriormean}. However, this approach
requires distributions that, in general, are not available in closed form.
Furthermore,~\eqref{eq:solution_marginal} is possibly a
high-dimensional optimization problem that does not admit an analytical
expression. To address this last problem, we use the
EM method to derive an iterative algorithm that
solves~\eqref{eq:solution_marginal}. We start by rewriting the marginal
likelihood as
\begin{equation}
  \p(y,v\,;\tau) = \int \p(y,v,g,w\,;\tau) \dif g \dif w.
\end{equation}
With this observation, we can see~\eqref{eq:solution_marginal} as a
maximum likelihood problem with latent variables, where the latent variables are
$g$ and $w$. Appealing to the theory of the EM method, we have that iterating
the two steps
\begin{description}
  \item[E-step:] Given an estimate $\hat \tau^{(k)}$ of $\tau$, construct the
    following lower bound of the marginal likelihood
    \begin{equation}\label{eq:Qtrue}
      \hspace{-1em} Q(\tau,\hat \tau^{(k)})\! = \!\int \!\log \p(y,v,g,w\,;\tau)
      \p(g,w|y,v\,;\hat \tau^{(k)})\dif g \dif w;
    \end{equation}
  \item[M-step:] Update the hyperparameter estimates as
    \begin{equation}
      \hat \tau^{(k+1)} = \arg\max_{\tau} Q(\tau,\hat \tau^{(k)});
    \end{equation}
\end{description}
from an arbitrary initial condition $\hat \tau^{(0)}$,
we obtain a sequence of estimates $\{\hat \tau^{(k)}\}$ of increasing
likelihood, which converges to a stationary point of the marginal likelihood of
the data. In practice, this stationary point will always be a local maximum:
saddle points are numerically unstable and minimal perturbations will drive the
sequence of updates away from them~\cite{mclachlan2007algorithm}.

Using the EM method, we have transformed the problem of maximizing the marginal
likelihood into a sequence of optimization problems. The whole point of the EM method is that
these problems should be simpler to solve than the original optimization
problem.

In addition to using the EM method to solve the marginal likelihood problem, we
can rewrite~\eqref{eq:solution_posteriormean} as
\begin{equation}\label{eq:posteriormeans}
   \hat g := \int g\, \p(g,w|y,v\,;\hat\tau) \dif w\dif g,\qquad
   \hat w := \int w\, \p(g,w|y,v\,;\hat\tau) \dif g\dif w.
\end{equation}
Comparing~\eqref{eq:posteriormeans} and the $Q$ function in the E-step, we see
that the solution of Problem~\ref{prob:sysid} using the procedure we have
described depends on expectations with respect to
the distribution $\p(g,w\,| y,v\,;\tau)$. This
distribution is, in general, not available in closed form. In the next section
we present three special cases when this distribution can be computed in closed form
and we present the resulting estimation algorithms. In
Section~\ref{sec:approximations}, we show two different ways to approximate
this joint posterior distribution in the general case.

\section{Cases with degenerate prior distributions}
There are cases where the integrals~\eqref{eq:Qtrue}
and~\eqref{eq:posteriormeans}, required to estimate uncertain-input
systems, admit closed-form solutions. This happens when either the prior for $g$
or for $w$ (or both) are degenerate distributions. This means that,
symbolically, we let the covariances $K_g(\rho)$ and $K_w(\theta)$ go to zero
and, respectively,
\begin{equation}
  \p(g;\rho) \to \delta\del{g-\mu_g(\rho)}, \quad
  \p(w;\theta) \to \delta\del{w-\mu_w(\theta)}.
\end{equation}
From these expressions, we see that the models of the unknown quantities $g$ and
$w$ are uniquely determined by the parameter vector $\tau$ (there is no
uncertainty or variability): therefore, we refer to these kind of models as
\emph{parametric models}. We now present three cases of parametric models that
admit closed form expressions for the EM algorithm.

\subsection{Semiparametric model}
The first model is called \emph{semiparametric}. It is obtained when
$K_w(\theta) \to 0$. This effectively means that the prior
density~\eqref{eq:model_w} collapses into the Dirac density centered around the
mean function, and the posterior distributions of the unknowns admit closed
form expressions:

\begin{lem}\label{lem:pigs_posterior_w}
  Consider the uncertain-input system~\eqref{eq:ui_model}. In the
  limit when $K_w(\theta)\to 0$, we have that
$\p(w|y,v;\tau) = \delta\del{w\!-\!\mu_w(\theta)}$.
\end{lem}
\begin{proof}\label{pf:pigs_posterior_w}
  When $K_w(\theta)\to 0$, the prior density becomes the degenerate normal
  distribution $\p(w;\theta) = \delta\big(w-\mu_w(\theta)\big)$.
  From the law of conditional expectation, we have
  \begin{equation}\label{eq:pf_pigs_posterior_w}
    \p(w|y,v;\tau) = \frac{\p(y,v|w;\tau)\delta(w - \mu_w(\theta))}{\p(y,v;\tau)};
\end{equation}
in addition, the evidence becomes
\begin{equation}
  \p(y,v;\tau) = \int \p(y,v|w;\tau)\p(w;\theta)\dif w =
  p(y,v|\mu_w(\theta);\tau).
\end{equation}
Plugging this expression into~\eqref{eq:pf_pigs_posterior_w} we have the
result.
\end{proof}

\begin{lem}\label{lem:pigs_posterior_g}
  Consider the uncertain-input system~\eqref{eq:ui_model}. In the limit when
  $K_w(\theta)\to 0$, the posterior distribution $\p(g|y,v;\tau)$ is Gaussian with
  covariance matrix and mean vector given by
  \begin{equation}\label{eq:pigs_posterior_g}
    P_g   = \del{\frac{1}{\sigma_y^2}{M_w(\theta)}^T M_w(\theta) +
    {K_g(\rho)}^{-1}}^{-1},\qquad
    m_g = P_g\del{\frac{1}{\sigma_y^2}{M_w(\theta)}^T
  y+{K_g(\rho)}^{-1}\mu_g(\rho)},
\end{equation}
where $M_w(\theta):=\T_{N\times N}\big(\mu_w(\theta)\big)$.
\begin{equation}
\end{equation}
\end{lem}
\begin{proof}
  Note that $y|g,w,v;\tau$ is an affine transformation of the 
  Gaussian random variable $\varepsilon$; hence, it is Gaussian. By
  the law of conditional expectation and ignoring terms independent of $g$, we
  have that
  \begin{equation}
    \begin{aligned}
      \log\, \p(g|w,y,v;\tau) &\cong \log p(y|g,w;\tau) + \log p(g;\rho)
                            \cong -\frac{1}{2\sigma_y^2}\enVert{y - M_w(\theta)g}^2 - \frac{1}{2}\enVert{g}^2_{{K_g(\rho)}^{-1}}\\
                            &\cong -\frac{1}{2} \enVert{g}_{P_g^{-1}} + g^T m_g\cong - \frac{1}{2}\enVert{g - m_g}^2_{{P_g}^{-1}}
    \end{aligned}
  \end{equation}
  where $P_g$ and $m_g$ are defined in~\eqref{eq:pigs_posterior_g}.
  Because it is quadratic, the posterior distribution of $g$ is Gaussian,
  with the indicated covariance matrix and mean vector.
\end{proof}

Thanks to Lemma~\ref{lem:pigs_posterior_w} and
Lemma~\ref{lem:pigs_posterior_g}, the E-step can be computed analytically when
$K_w(\theta) = 0$, and
the function $Q(\tau, \hat \tau^{(k)})$ admits a closed-form expression. To
this end,
let $\Delta^{k}$ be the $N$ by $N$ matrix given by
\begin{equation}
  {[\Delta^k]}_{i,j} =
  \begin{cases}{}
    1 & \text{if } i+j-1 = k\\
    0 & \text{otherwise}
  \end{cases}
\end{equation}
and let
\begin{equation}\label{eq:R}
  R = \begin{bmatrix}
    \Delta^1 & \Delta^2 & \Delta^3 & \cdots & \Delta^N
  \end{bmatrix}.
\end{equation}
Then, we have the following result.

\begin{thm}\label{thm:Qpigs}
  Consider a semiparametric uncertain-input model with $K_w(\theta) = 0$.
  Let $\hat \tau^{(k)}$ be estimates of the hyperparameters at the $k$th
  iteration of the EM method and let $\hat g^{(k)}$ and $\hat P_g^{(k)}$ be the moments
  in~\eqref{eq:pigs_posterior_g} when $\tau=\hat \tau^{(k)}$. Define
  \begin{equation}
    \begin{aligned}
    \hat R_y(\theta) &= \enVert[1]{ y - \hat G^{(k)}\mu_w(\theta) }^2,\\
    \hat R_v(\theta) &= \enVert[1]{v - \mu_w(\theta)}^2,\\
    \hat S_g^{(k)} &=  R^T \del[1]{I_N \otimes \hat P_g^{(k)}} R;
    \end{aligned}
  \end{equation}
  where $\hat G^{(k)} = T_{N\times N}(\hat g^{(k)})$. 
  Then, the function $Q(\tau,\hat \tau^{(k)})$ is
  given by
   \begin{equation}\label{eq:Qpigs}
     \begin{aligned}
       Q(\tau,\hat \tau^{(k)}) &= -\frac{1}{2\sigma_v^2}\hat R_v(\theta) - \frac{N}{2}\log \sigma_v^2
     -\frac{1}{2\sigma_y^2}\del{\hat R_y(\theta) + \enVert{\mu_w(\theta)}_{
     \hat S_g^{(k)}}^2} 
     - \frac{N}{2}\log
     \sigma_y^2\\
     &- \frac{1}{2} \enVert{\hat g^{(k)} - \mu_g(\rho)}^2_{{K_g(\rho)}^{-1}} - \frac{1}{2} \trace\cbr{{K_g(\rho)}^{-1}\hat P_g^{(k)}} -
     \frac{1}{2}\log\det K_g(\rho).
     \end{aligned}
   \end{equation}
\end{thm}
\begin{proof}
  See Appendix~\ref{pf:Qpigs}.
\end{proof}

From~\eqref{eq:Qpigs}, we see that the optimization with respect to $\theta$ is
not independent of $\sigma_y^2$ and $\sigma_v^2$. Therefore, to update the
hyperparameter $\theta$, we use a conditional-maximization
step~\cite{meng1993maximum}, where we keep the noise variances fixed
to their values at the previous iterations. The use of the conditional
maximization step allows us to write the updates of the EM method in closed
form:

\begin{cor}\label{cor:Qpigs}
  At the $k$th iteration of the EM method, the parameters can be updated as
  \begin{equation}
  \begin{aligned}
    \hat \rho^{(k+1)} &= \arg\min_{\rho}
  \trace\cbr[1]{{K_g(\rho)}^{-1}\hat P_g^{(k)}} + \log \det K_g(\rho)\\
  &+ \enVert[1]{\hat g^{(k)}-\mu_g(\rho)}^2_{K_g^{-1}(\rho)},\\
  \hat\theta^{(k+1)}& = \arg \min_\theta\!
  \frac{\hat R_v(\theta)}{2\hat \sigma_v^{(k)\,2}} + 
  \frac{1}{2\hat\sigma_y^{(k)\,2}}\!\del[2]{\!
    \hat R_y(\theta)\!+\! \enVert{\mu_w(\theta)}^2_{
    \hat S_g^{(k)}}\!}\!,\\
    \hat \sigma_y^{(k+1)} &= \frac{1}{N}\del{\hat R_y(\hat \theta^{(k+1)}) +
    \enVert[1]{\mu_w(\hat \theta^{(k+1)})}^2_{S_g^{(k)}}},\\
    \hat \sigma_v^{(k+1)} &= \frac{1}{N}\hat R_v(\hat \theta^{(k+1)})\,.
\end{aligned}
\end{equation}
\end{cor}
\begin{proof}
  Follows from the two-step maximization of~\eqref{eq:Qpigs}: first, maximize
  with respect to $\theta$ and $\rho$ keeping $\sigma_y^2$ and $\sigma_v^2$
  fixed to their values at the previous iteration; then, maximize with respect
  to $\sigma_y^2$ and $\sigma_v^2$ using the updated values of the
  hyperparameters. 
\end{proof}

Thanks to Corollary~\ref{cor:Qpigs}, we have a simple way to compute the EM
estimates of the kernel hyperparameters and of the noise variances for
semiparametric models with $K_w(\theta) = 0$: starting
from an initial value of the unknown parameters, we first update the
hyperparameters $\rho$ and
$\theta$, then we use the new values to update the noise variances $\sigma_y^2$
and $\sigma_v^2$. Under mild regularity conditions, this procedure yields a
sequence of estimates that converges to a local maximum of the marginal
likelihood~(it is a Generalized EM sequence, see~\cite{wu1983convergence}).

\begin{rem}
  In this section, we have presented the case when $K_w(\theta)\to 0$. However,
  thanks to the symmetry of the model assured by~\eqref{eq:toeplitz_property},
  the same kind of algorithm works when $K_g(\rho)\to 0$ (by exchanging the
  roles of $g$ and $w$).
\end{rem}

\subsection{Parametric model}
In case we let both the input and the system covariance matrices go to zero,
all the variability in the model is removed, and we are left with classical
\emph{parametric} models. In this case, the marginal likelihood of the data collapses
into the likelihood where the impulse response and the input are
replaced with the parametric models $\mu_g(\rho)$ and $\mu_w(\theta)$
\begin{equation}
  \p(y,v;\rho,\theta,\sigma^2) = \int\!\!
  \p(y,v|g,w;\sigma^2)\p(g;\rho)\p(w;\theta)\dif g\dif w
  =\p(y,v|\,g\!=\!\mu_g(\rho),\,w\!=\!\mu_w(\theta);\,\sigma^2).
\end{equation}
In other words, the marginal likelihood of the data is the distribution of the
data conditioned on the events $g=\mu_g(\rho)$ and $w=\mu_w(\theta)$. This
distribution is given in closed form by
\begin{equation}\label{eq:Qparametric}
  \log \p(y,v|\mu_g(\rho),\mu_w(\theta);\sigma^2) =
  -\frac{1}{2\sigma_y^2}\enVert{y-M_w(\theta)\mu_g(\rho)}^2
  -\frac{N}{2}\log\sigma_y^2
  -\frac{1}{2\sigma_v^2}\enVert{v-\mu_w(\theta)}^2-\frac{N}{2}\log \sigma_v^2.
\end{equation}
where $M_w(\theta)$ is the Toeplitz matrix of $\mu_w(\theta)$.

In this parametric-model case, we have that the posterior means reduce to the
prior means and the maximum marginal-likelihood criterion collapses into the
classical maximum-likelihood or prediction-error estimation method. To
estimate the system, we first maximize~\eqref{eq:Qparametric} to find the
parameter values $\hat \tau$; then, we estimate the system with
\begin{equation}
  \hat g = \mu_g(\hat\rho),\quad \hat w=\mu_w(\hat \theta).
\end{equation}

The strategy to maximize~\eqref{eq:Qparametric} depends on the specific
structure of the problem. In some applications, concentrated-likelihood or
integrated-likelihood approaches have been proposed (for a review,
see~\cite{berger1999integrated}). An interesting consistent approach, for the
parametric EIV case, has been proposed in~\cite{zhang2015errors}.
In~\cite{bai2004convergence}, the authors show that if $g$ and $w$ are linearly
parameterized, alternating between estimation of $g$ and of $w$ leads to the
minimum of~\eqref{eq:Qparametric}.

\begin{rem}
  The EM based algorithm presented in Section~\ref{sec:estimation} cannot be
  used in the parametric model case because of the impulsive posterior
  distributions: during the M-step, the method is overconfident in the current
  value of the parameters and no update occurs. However, the EM method can be
  used in the parametric case by considering a covariance matrix that shrinks
  toward zero at every iteration.
\end{rem}

\section{Approximations of the joint posterior
distribution}\label{sec:approximations}
In the previous section, we have shown three cases in which the collapse of the prior
distribution allows us to express the marginal likelihood of the data and the
posterior distributions in closed form. In general, however, these distributions do not
have a closed form expression. Therefore, in this section, we present two ways to
approximate the joint posterior distribution
$\p(g,w\,| y,v\,;\tau)$.
In the first, we make a particle approximation. The particles are drawn
from the joint posterior using an MCMC method. In the second, we make a
variational approximation of the joint posterior.

\subsection{Markov Chain Monte Carlo integration}\label{sec:mcmc}
Monte Carlo methods are built around the concept of \emph{particle
approximation}. In a particle approximation method, a density with a complicated
functional form is approximated with a set of point probabilities---that is, we
approximate a density $\p(x)$ according to
\begin{equation}
  \p(x) \approx \frac{1}{M}\sum_{j=1}^M \delta(x-x_j).
\end{equation}
If the particle locations $x_j$ are drawn from $\p(x)$, and the number of
particles $M$ is large enough, the expectation of any measurable function $f(x)$
over any set can be approximated as
\begin{equation}\label{eq:monte_carlo_integral}
  \E\{f(x)\} = \int\!f(x)\,\p(x) \dif x \approx \frac{1}{M}\sum_{j=1}^M f(x_j),
\end{equation}
where $\{x_j\}$ are drawn from  $\p(x)$. This result comes directly from the sampling
property of the Dirac density $\delta(\,\cdot\,)$. From a different
perspective, we can see~\eqref{eq:monte_carlo_integral} as an estimation of the
true expectation. With this interpretation, we have that this estimator is
unbiased,
\begin{equation}
  \E\Bigg\{ \frac{1}{M}\sum_{j=1}^M
   f(x_j)\Bigg\} = \E\{f(x)\},
\end{equation}
and its covariance is inversely proportional to the number of samples used,
\begin{equation}
  \cov\Bigg\{ \frac{1}{M}\sum_{j=1}^M
  f(x_j)\Bigg\} = \frac{1}{M}\cov\{f(x)\}.
\end{equation}
In practice, the number of samples needed depends on the specific application: in
certain applications, few particles (say 10 or 20) may suffice; in other
applications, we might need a much larger number of particles (in the order of
thousands; for a complete treatment, see~\cite[Chapter~11]{bishop2006pattern}).

When implementing Monte Carlo integrations, a common approach is MCMC\@. In
these methods, we set up a Markov chain whose stationary distribution is the
distribution we want to approximate and we run it to
collect samples~\cite{gilks1996markov}.

One convenient way to create a Markov chain is \emph{Gibbs sampling}. Using
this method, we obtain a particle approximation of a joint distribution (called
the \emph{target distribution}) by
sampling from all the \emph{full conditional} distributions---the distribution of
one random variable conditioned on all other variables---in sequence. This
procedure results in a Markov chain that has the target distribution as its
stationary distribution. Contrary to many other sampling methods, Gibbs
sampling does not include a rejection step; this means that the samples
proposed at every step are accepted as samples from the chain. This may lead to
faster mixing and decorrelation of the chain compared to other MCMC
methods~\cite[Chapter~11]{bishop2006pattern}.

The main drawback with Gibbs sampling is that we must sample the full
conditional distributions of all variables. Therefore, it is only applicable if
these distributions have a functionally convenient form. In the case at hand,
we have the following results.
\begin{lem}\label{lem:cond_g} Consider the uncertain-input model~\eqref{eq:ui_model}. The
  density $\p(g|y,w;\tau)$ is Gaussian with covariance
  matrix and mean vector given by
\begin{equation}\label{eq:conditional_g_pars}
    P_g = {\left(\frac{W^T W}{\sigma_y^2} +
    {K_g(\rho)}^{-1}\right)}^{-1},\qquad
    m_g =  P_g\left(\frac{W^T y}{\sigma_y^2} +
  {K_g(\rho)}^{-1}\mu_g(\rho)\right).
\end{equation}
\end{lem}
\begin{proof}
  The proof follows along the same line of reasoning as the proof of Lemma~\ref{lem:pigs_posterior_g}.
\end{proof}

\begin{lem}\label{lem:cond_w} Consider the uncertain-input model~\eqref{eq:ui_model}. The
  density $\p(w|y,v,g;\tau)$ is Gaussian with covariance matrix
  and mean given by
\begin{equation}\label{eq:conditional_w_pars}
    P_w = {\left(\frac{G^T G}{\sigma_y^2} + \frac{I_N}{\sigma^2_v} +
  {K_w(\theta)}^{-1}\right)}^{-1},\qquad
 m_w =P_w\left(\frac{G^T y}{\sigma_y^2} + \frac{v}{\sigma_v^2} +
 {K_w(\theta)}^{-1}\mu_w(\theta)\right).
\end{equation}
\end{lem}
\begin{proof}
  Because $y$ and $v$ are conditionally independent given $w$ and $g$, we have
  that
  \begin{equation}
    \begin{aligned}
      &\log\p(w|y,v,g;\tau) \cong
      \log\del{\p(y|g,w;\sigma_y^2)\p(v|w;\sigma_v^2)\p(w;\theta)}\\
          &\cong -\frac{1}{2\sigma_y^2}\enVert{y \!-\! Gw}^2
          \!\!-\!\frac{1}{2\sigma_v^2}\enVert{v \!-\! w}^2\! \!-\! \frac
      {1}{2} \enVert{w \!-\! \mu_w(\theta)}^2_{{K_w(\theta)}^{-1}}\\
      &\cong -\frac{1}{2}\enVert{w}^2_{P_w^{-1}}
      + w^T m_w \cong -\frac{1}{2} \enVert{w - m_w}_{P_w^{-1}}^2
    \end{aligned}
  \end{equation}
  where $P_w$ and $m_w$ are given in~\eqref{eq:conditional_w_pars}. The
  log-density of $w|y,v,g,w;\tau$ is quadratic and, hence, it is Gaussian with
  the indicated mean vector and covariance matrix.
\end{proof}

\begin{rem}\label{rem:general_mc}
  In case we consider the more general Gaussian process
  model~\eqref{eq:gp_joint}, where $g$ and $w$ are a priori dependent,
  Lemma~\ref{lem:cond_g} and Lemma~\ref{lem:cond_w} still hold with slightly
  modified expressions for the mean vectors and covariance matrices (to account
  for the prior correlation). For instance, the conditional density of $g$ is
  Gaussian with covariance matrix and mean vector given by
  \begin{equation}
    P_g \!=\! {\left(\frac{W^T W}{\sigma_y^2} +
    \Lambda_g(\rho,\theta)\right)}^{-1},\qquad
    m_g \!=\!  P_g\!\left( \!\frac{W^T y }{\sigma_y^2} +
    \Lambda_g(\rho,\theta)\mu_g(\rho) + \Lambda_{gw}(\rho,\theta)(w \!-\!
  \mu_w(\theta))\!\right)\!,
\end{equation}
where $\Lambda_{gw}(\rho,\theta)$ and $\Lambda_{g}(\rho,\theta)$ are,
respectively, the lower left and right blocks of the inverse of the prior
covariance matrix.
\end{rem}

In view of Lemma~\ref{lem:cond_g} and Lemma~\ref{lem:cond_w}, we can easily
set up the Gibbs sampler to draw from the joint posterior distribution: from any
initialization of the impulse response $g^{(0)}$ and of the input signal $w^{(0)}$, we sample
\begin{equation}\label{eq:gibbs}
  \begin{aligned}
    g^{(j+1)}&|w^{(j)},y,v;\tau \sim \N(m_g^{(j)},P_g^{(j)}),\\
    w^{(j+1)}&|g^{(j+i)},y,v;\tau  \sim \N(m_w^{(j)},P_w^{(j)}).
  \end{aligned}
\end{equation}
where $m_g^{(j)}$ and $P_g^{(j)}$ are the mean and covariance
in~\eqref{eq:conditional_g_pars} when $w=w^{(j)}$, and where $m_w^{(j)}$ and
$P_w^{(j)}$ are the mean and covariance in~\eqref{eq:conditional_w_pars} when
$g=g^{(j+1)}$.

Because it is a Markov chain, the samples drawn using~\eqref{eq:gibbs} are
correlated, and subsequent samples have memory about the initial conditions and
are far away from the stationary distribution (which is equal to the target
distribution). Therefore, we discard the first samples of the Markov chain, and
we only retain the $M$ samples after a \emph{burn-in} of $B$ samples:
\begin{equation}\label{eq:burnin}
  \bar g^{(j)} = g^{(j+B)}, \quad
  \bar w^{(j)} = w^{(j+B)}, \quad j = 1,\ldots,M.
\end{equation}
If the burn-in is large enough, the Markov chain has lost its memory
about the initial conditions and is producing samples that come form the
stationary distribution. The choice of the length of the burn-in is a
difficult problem, and some heuristic algorithms have been
proposed (see~\cite[Section~1.4.6]{gilks1996markov}).

When we have drawn enough samples from
the Markov chain, we compute the Monte Carlo estimate of the function $Q$; in
other words, we replace the E-step in the EM method with a Monte Carlo E-step
(this is sometimes known as the MCEM method; see~\cite{wei1990monte}). We create the approximate lower
bound (at the $k$th iteration of the EM method) by setting
\begin{equation}
  Q^{\mc}(\tau,\hat \tau^{(k)}) = \frac{1}{M_k} \sum_{j=1}^{M_k} \log \p(y,v,\bar g^{(j,k)},
    \bar w^{(j,k)};\tau)
\end{equation}
where $\bar g^{(j,k)}$ and $\bar w^{(j,k)}$ are samples from the stationary
distribution of~\eqref{eq:gibbs} at the $k$th iteration of the EM method. In
the uncertain-input case, the function $Q^{\mc}$ is available in closed form as
a function of the sample moments of $g$ and $w$.
\begin{thm}\label{thm:Qmc}
  Let $\cbr[1]{\bar g^{(j,k)}}_{j=1}^{M_k}$ and
  $\cbr[1]{\bar w^{(j,k)}}_{j=1}^{M_k}$ be
  samples from the stationary distribution of the Gibbs
  sampler~\eqref{eq:gibbs} at the $k$th
  iteration of the EM method and define
  \begin{equation}\label{eq:mcem_moments}
    \begin{aligned}
      \hat g^{(k)} &= \frac{1}{M_k}\sum_{j=1}^{M_k} \bar g^{(j,k)},\quad \hat w^{(k)}
      = \frac{1}{M_k}\sum_{j=1}^{M_k} \bar w^{(j,k)},\\
      \hat P_g^{(k)} &= \frac{1}{M_k}\sum_{j=1}^{M_k} \del{\bar g^{(j,k)} -\hat g^{(k)}}\del{\bar g^{(j,k)} - \hat g^{(k)}}^T\\
      \hat P_w^{(k)}  &= \frac{1}{M_k}\sum_{j=1}^{M_k} \del{\bar w^{(j,k)} - \hat w^{(k)}}\del{
      \bar w^{(j,k)} - \hat w^{(k)}}^T\\
      \hat R_v^{(k)} &= \frac{1}{M_k}\sum_{j=1}^{M_k} \enVert{v - \bar w^{(j,k)}}^2,\\
      \hat R_y^{(k)} &= \frac{1}{M_k}\sum_{j=1}^{M_k} \enVert{y - \bar G^{(j,k)}\bar w^{(j,k)}}^2.
\end{aligned}
  \end{equation}
  Then, the function $Q^{\mc}(\tau, \hat \tau^{(k)})$ is given by
  \begin{equation}\label{eq:Qmc}
    \mathmakebox[0.88\columnwidth][l]{\begin{aligned}
        \!Q^{\mc}(\tau,\hat \tau^{(k)}) &= \!- \frac{\hat R_v^{(k)}}{2\sigma_v^2}
      \!-\! \frac{N}{2} \log
      \sigma_v^2 \!-\!\frac{\hat R_y^{(k)}}{2\sigma_y^2} \!-\! \frac{N}{2}\log \sigma_y^2
      \!-\! \frac{1}{2}\trace\cbr{{K_g(\rho)}^{-1} \hat P_g^{(k)}}\\
      \!&-\!\frac{1}{2} \enVert{\hat g^{(k)} \!-\!
    \mu_g(\rho)}^{2}_{{K_g(\rho)}^{-1}}
  \!-\!\frac{1}{2} \trace\cbr{{K_w(\theta)}^{-1} \hat P_w^{(k)}}
    \!-\!\frac{1}{2} \enVert{\hat w^{(k)} \!-\!
  \mu_w(\theta)}^{2}_{{K_w(\theta)}^{-1}}\\ 
  &\!-\!\frac{1}{2} \log \det K_g(\rho) \!-\!\frac{1}{2} \log \det K_w(\theta).
    \end{aligned}}
  \end{equation}
\end{thm}
\begin{proof}
  See Appendix~\ref{pf:Qmc}.
\end{proof}

In the M-step, we update the hyperparameters $\hat \tau^{(k)}$ by maximizing
the approximate lower bound of the marginal likelihood, $Q^{\mc}$. Because of the closed
form expression in Theorem~\ref{thm:Qmc}, the M-step splits into the decoupled
optimization problems for the kernel hyperparameters and the noise variances
according to the following:
\begin{cor}\label{cor:Qmc}
  At the $k$th iteration of the EM method, the kernel hyperparameters can be
  updated as
  \begin{equation}
    \begin{aligned}
      \hat \rho^{(k+1)} &= \arg\min_\rho
      \enVert{\hat g^{(k)}-\mu_g(\rho)}^{2}_{{K_g(\rho)}^{-1}}
      + \trace\cbr{{K_g(\rho)}^{-1} \hat P^{(k)}_g} + \log \det K_g(\rho),\\
      \hat \theta^{(k+1)} &= \arg\min_\theta
      \enVert{\hat w^{(k)}-\mu_w(\theta)}^{2}_{{K_w(\theta)}^{-1}}
      + \trace\cbr{{K_w(\theta)}^{-1} \hat P^{(k)}_w} + \log \det K_w(\theta),
    \end{aligned}
  \end{equation}
  and  the noise variances can be updated as
  \begin{equation}
    \hat \sigma_v^{2\,(k+1)} = \frac{\hat R_v^{(k)}}{N},\qquad
    \hat \sigma_y^{2\,(k+1)} = \frac{\hat R_y^{(k)}}{N}.
  \end{equation}
\end{cor}
\begin{proof}
  Follows from direct maximization of~\eqref{eq:Qmc}.
\end{proof}

Thanks to Theorem~\ref{thm:Qmc} and Corollary~\ref{cor:Qmc}, we have a simple
way to compute the MCEM estimates of the kernel hyperparameters and of the
noise variances; starting from an initial value of the hyperparameters, we iterate
the following three steps:
\begin{enumerate}
  \item Run a Gibbs sampler according to~\eqref{eq:gibbs}.
  \item Collect the samples according to~\eqref{eq:burnin} and compute the moments according
    to~\eqref{eq:mcem_moments}.
  \item Update the parameters according to Corollary~\ref{cor:Qmc}.
\end{enumerate}
Under mild regularity conditions, these iterations yield a sequence of
parameter estimates that converges to a stationary point of the marginal likelihood of the
data (under the condition that the number of particles $M_k$ at iteration $k$ is
such that that $\sum_{k=1}^\infty M_k^{-1} = \infty$;
see~\cite{neath2013convergence}). Then, using the estimated
hyperparameters, we can run a new Gibbs sampler and approximate the integrals
in~\eqref{eq:posteriormeans} with averages over the samples:
\begin{equation}\label{eq:particle_means}
  \hat g \approx \frac{1}{M} \sum_{j=1}^M \bar g^{(j)},\quad
  \hat w \approx \frac{1}{M} \sum_{j=1}^M \bar w^{(j)}.
\end{equation}

\subsection{Variational Bayes approximation}\label{sec:vb}
The second  method we present is a variational approximation method.
Instead of approximating the unknown joint posterior density using sampling, we
propose an analytically tractable family of distributions and we look for the best
approximation of the unknown posterior density within that family.

The variational Bayes method hinges on the fact that
\begin{equation}
  \log\p(y,v,g,w;\tau) = \log(g,w|y,v;\tau) + \log\p(y,v;\tau).
\end{equation}
Hence, for any proposal distribution $\q$ in some family of
distributions $\Q$, we can write
\begin{equation}
  \log \p(y,v;\tau) = \log \frac{\p(y,v,g,w;\tau)}{\q(g,w)} - \log
  \frac{\p(g,w|y,v;\tau)}{\q(g,w)}.
\end{equation}
Taking the expectation with respect to $\q$ and observing that the left hand
side is independent of $g$ and $w$, we get that
\begin{equation}\label{eq:marginal_split}
  \log \p(y,v;\tau) = L(\q) + KL(\q),
\end{equation}
where we have defined the functional
\begin{equation}
  L(\q) = \int \log\left(\frac{\p(y,v,g,w;\tau)}{\q(g,w)}\right)\q(g,w)\dif
  g \dif w,
\end{equation}
and the \emph{Kullback-Leibler (KL) distance}~\cite{kullback1951information}
\begin{equation}
  KL(\q) = \int \log\left(\frac{\q(g,w)}{\p(g,w|y,v;\tau)}\right)\q(g,w)\dif
  g \dif w.
\end{equation}
Although the KL distance is not a metric---it is not symmetric
and it does not satisfy the triangle inequality---it is a useful measure of
similarity between probability
distributions (see~\cite[Section~1.6.1]{bishop2006pattern}).

Because the left hand side of~\eqref{eq:marginal_split} is independent of $\q$,
we can find the distribution $\q^\star$ with minimum distance
(in the KL sense) to the target distribution by maximizing
the functional $L(\q)$ with respect to $\q\in \Q$,
\begin{equation}\label{eq:kl_optimization}
  \q^\star(g,w) =\arg\min_{\q \in \Q} KL(\q) = \arg\max_{\q\in \Q} L(\q).
\end{equation}
This technique allows us to use the known functional $L(\q)$ to find the $\q$
with minimum KL distance to the unknown joint posterior distribution.

To use the variational approximation, we need to fix a family of distributions
$\Q$ among which to look for $\q^\star$. In this work, we use a
\emph{mean-field approximation}, meaning that we look for an approximation of
the posterior distribution where $g$ and $w$ are independent given the data; in
other words we consider proposal distributions that factorize into two
independent factors according to
\begin{equation}
  \q(g,w) = \q_g(g)\q_w(w).
\end{equation}

After choosing the family of proposal distributions, we need to find the best
approximation $\q^\star$ in terms of KL distance to the unknown
posterior distribution; in view of~\eqref{eq:kl_optimization}, the solution is
given by
\begin{equation}
  \q^\star(g,w) = \arg\max_{\q_g,\,\q_w} L(\q_g\q_w).
\end{equation}
Consider first the factor $\q_g$. We have that
\begin{equation}
  \begin{aligned}
    L&(\q_g\q_w) = \int
  \log\left(\frac{\p(y,v,g,w;\tau)}{\q_g(g)\q_w(w)}\right)
  \q_g(g) \q_w(w)\dif g \dif w,\\
  &\!\cong\!\int \!\sbr{\log\p(y,v,g,w;\tau)\q_w(w)\dif w - \log \q_g(g)
}\q_g(g)\dif g,
  \end{aligned}
\end{equation}
ignoring terms independent of $\q_g$.
If we define the distribution $\p_w(y,v,g;\tau)$ such that
\begin{equation}
  \log \p_w(y,v,g;\tau) = \int \log \p(y,v,g,w;\tau) \q_w(w)\dif w,
\end{equation}
we have that, again ignoring terms independent of $\q_g(g)$,
\begin{equation}
  L(\q_g\q_w) \cong - \int \log\del{\frac{\p_w(y,v,g;\tau)}{\q_g(g)}}\q_g(g)
  \dif g,
\end{equation}
which is the negative KL distance between the factor $\q_g$ and the
density $\p_w(y,v,g;\tau)$. Because the KL distance is nonnegative,
by choosing $\q_g^\star(g)=\p_w(y,v,g;\tau)$ (where the KL distance is zero)
we are maximizing the functional $L$ with respect to
$\q_g$. Considering now $\q_w(w)$, we can trace the same argument and find that
the optimal choice is
\begin{equation}\label{eq:q_w_star}
  \log \q^\star_w(w) = \int \log p(y,v,g,w;\tau)\q^\star_g(g)\dif g,
\end{equation}
where $\q^\star_g(g)$ is the solution of
\begin{equation}\label{eq:q_g_star}
  \log \q^\star_g(g) = \int \log p(y,v,g,w;\tau)\q^\star_w(w)\dif w.
\end{equation}
The maximum of $L(\q_g\q_w)$ is, therefore, the simultaneous solution
of~\eqref{eq:q_w_star} and~\eqref{eq:q_g_star}. The solution can be found with
the following iterative procedure: from an initialization
$\q^{(0)}_g$ and $\q_{w}^{(0)}$ of the densities, compute
\begin{equation}\label{eq:variational_bayes}
  \begin{aligned}
    \!\log \q^{(j+1)}_w(w) &= \int \log p(y,v,g,w;\tau)\q^{(j)}_g(g)\dif g,\\
    \!\log \q^{(j+1)}_g(g) &= \int \log p(y,v,g,w;\tau)\q^{(j+1)}_w(w)\dif w.
  \end{aligned}
\end{equation}
This iterative procedure will converge to the
simultaneous solution of~\eqref{eq:q_g_star} and~\eqref{eq:q_w_star}
(see~\cite[Chapter~10]{bishop2006pattern}; see also~\cite{boyd2004convex}).

As was the case for the Gibbs sampler, which can be used only if it easy to
sample from the full conditional distributions, the variational approximation of
the joint posterior is only useful if it is possible to compute the expectations
in~\eqref{eq:q_w_star} and~\eqref{eq:q_g_star}. In the uncertain-input case, we
have the following result.
\begin{thm}\label{thm:vb_gaussian}
  Let $\q_g^\star\q^\star_w$ be the factorized density with minimum KL distance
  to posterior density $\p(g,w|y,v;\tau)$, for a fixed value of the
  hyperparameters. Then, $\q^\star_g$ and $\q^\star_w$ are Gaussian distributions.
\end{thm}
\begin{proof}
  See Appendix~\ref{pf:vb_gaussian}.
\end{proof}

Theorem~\ref{thm:vb_gaussian} allows us to compute expectations with respect to
$\q^\star_g$ and $\q^\star_w$ easily. In addition, at every iteration
of~\eqref{eq:variational_bayes} the approximating densities remain Gaussian.
This allows us to write the update~\eqref{eq:variational_bayes} in terms
of the first and second moments of the approximating densities:
\begin{cor}\label{cor:iterative_vb}
  Let $ w^{(j)}$ and $ g^{(j)}$ be the mean vectors of $\q_w^{(j)}$
  and $\q_g^{(j)}$ at the $j$th iteration of~\eqref{eq:variational_bayes}
  and let $ P_w^{(j)}$ and $ P_g^{(j)}$ be the covariance matrices. Let
  $ g^{(j+1)}$, $ w^{(j+1)}$, $ P^{(j+1)}_g$, and
  $ P_w^{(j+1)}$ be the mean vectors and covariance matrices at the
  $(j+1)$th iteration. Let
  \begin{equation}
    \begin{aligned}
      T_g^{(j)}&= R\del{I_N \otimes \sbr{ P_g^{(j)}+  g^{(j)}
       g^{(j)\,^T}}}R^T,\\
      T_w^{(j+1)}&= R\del{I_N \otimes \sbr{ P_w^{(j+1)}+  w^{(j+1)}
           w^{(j+1)\,T}}}R^T.
    \end{aligned}
  \end{equation}
  where the matrix $R$ is defined in~\eqref{eq:R}. Then,
  \begin{equation}\label{eq:vb_moments}
    \mathmakebox[0.88\columnwidth][l]{
  \begin{aligned}
   \!P_w^{(j+1)}&\!=\!\del[3]{\frac{1}{\sigma_y^2}T_g^{(j)} + \frac{1}{\sigma^2_v}I_n +
    {K_w(\theta)}^{-1}}^{-1},\\
   \!w^{(j+1)} &\!=\! P_w^{(j+1)}\!\del[3]{\frac{ G^{(j)\,T}}{\sigma_y^2}y + \frac{1}{\sigma_v^2}v +
  {K_w(\theta)}^{-1}\mu_w(\theta)}\!,\\
     \!P_g^{(j+1)} &= \del[3]{\frac{1}{\sigma_y^2}T_w^{(j+1)} +
      {K_g(\rho)}^{-1}}^{-1},\\
     \!g^{(j+1)} &=  P_g^{(j+1)}\del[3]{\frac{ W^{(j+1)\,T}}{\sigma_y^2}y +
  {K_g(\rho)}^{-1}\mu_g(\rho)}.
  \end{aligned}}
\end{equation}
\end{cor}
\begin{proof}
  See Appendix~\ref{pf:iterative_vb}.
\end{proof}

Thanks to Corollary~\ref{cor:iterative_vb}, we can iteratively update the
moments of the Gaussian factors, and the iterations will converge to the moments
of optimal variational approximation of the joint posterior distribution.

\begin{rem}
  In case we consider the more general Gaussian process
  model~\eqref{eq:gp_joint}, the results of
  Theorem~\ref{thm:vb_gaussian}  and of Corollary~\ref{cor:iterative_vb} still
  hold with minor modifications (similarly to what is presented in
  Remark~\ref{rem:general_mc}). However, the approximation of posterior
  independence may not make sense when using a-priori dependent Gaussian
  process models.
\end{rem}

Using the factorized approximation of the joint distribution,
we can approximate the E-step in the EM method with a variational
E-step (this is sometimes known as the VBEM method,
see~\cite{beal2003variational}). We
create the variational approximation of the lower bound (at the $k$th iteration
of the EM method) by setting
\begin{equation}
  Q^{\vb}(\tau,\hat\tau^{(k)}) \!:=\! \int\!
  \log\p(y,v,g,w;\tau)\hat \q_g^{(k)}(g)\hat \q_w^{(k)}(w)\dif w\dif g,
\end{equation}
where $\hat \q_g^{(k)}$ and $\hat\q_w^{(k)}$ are the limits of the variational
Bayes iterations with the hyperparameters set to $\hat \tau^{(k)}$.

Because the complete-data likelihood is quadratic in $g$ and $w$, the
approximation $Q^\vb$ admits the closed form expression
in function of the moments of $g$ and $w$.
\begin{thm}\label{thm:Qvb}
  Let $\hat g^{(k)}$ and $\hat w^{(k)}$ be the mean vectors of $\hat\q_g^{(k)}$
  and of $\hat \q_w^{(k)}$, respectively, and let $\hat P^{(k)}$ and
  $\hat P^{(k)}$  be their covariance matrices. Define
  \begin{equation}
    \begin{aligned}
      \hat S_w^{(k)} &= R \del[1]{I_n \otimes \hat P_g^{(k)}} R^T,
       &\hat T_w^{(k)} &= \hat S_w^{(k)} \!+\! \hat W^{(k)T}\hat W^{(k)},\\
      \hat R_v^{(k)}  &= \enVert[1]{v - \hat w^{(k)}}^2, 
      &\hat R_y^{(k)} &= \enVert[1]{y - \hat W^{(k)} \hat g^{(k)}}^2,
    \end{aligned}
  \end{equation}
  where $R$ is defined in~\eqref{eq:R}.
  Then,
  \begin{equation}\label{eq:Qvb}
    \begin{aligned}
      &Q^\vb(\tau,\hat \tau^{(k)}) = -\frac{\hat R_v^{(k)}}{2\sigma_v^2}
      -\frac{N}{2}\log \sigma_v^2 - \frac{N}{2}\log \sigma_y^2- \frac{1}{2\sigma_y^2}\del{\hat R_y^{(k)} -
    \enVert[1]{\hat g^{(k)}}^2_{\hat S_w^{(k)}} - \trace\cbr{\hat T_w^{(k)} P_g^{(k)}}}\\
      &- \frac{1}{2}\trace\cbr{{K_g(\rho)}^{-1} \hat P_g^{(k)}} -\frac{1}{2} \enVert{\hat g^{(k)} -
    \mu_g(\rho)}^{2}_{{K_g(\rho)}^{-1}}-\frac{1}{2}
    \trace\cbr{{K_w(\theta)}^{-1} \hat P_w^{(k)}}\\
    &-\frac{1}{2} \enVert{\hat w^{(k)} -
  \mu_w(\theta)}^{2}_{{K_w(\theta)}^{-1}}-\frac{1}{2} \log \det K_g(\rho)
  -\frac{1}{2} \log \det K_w(\theta).
    \end{aligned}
  \end{equation}
\end{thm}
\begin{proof}
  See Appendix~\ref{pf:Qvb}.
\end{proof}

Thanks to the structure of the function $Q^\vb(\tau,\hat \tau^{(k)})$, 
the M-step splits into
decoupled optimization problems for the kernel hyperparameters and for the
noise variances.
\begin{cor}\label{cor:Qvb}
  At the $k$th iteration of the EM method, the kernel hyperparameters can be
  updated as
\begin{equation}\label{eq:vbem_update_rho}
    \begin{aligned}
      \hat \rho^{(k+1)} &= \arg\min_\rho
      \enVert{\hat g^{(k)}-\mu_g(\rho)}^{2}_{{K_g(\rho)}^{-1}}+
      \trace\cbr{{K_g(\rho)}^{-1} \hat P^{(k)}_g} + \log \det K_g(\rho),\\
      \hat \theta^{(k+1)} &= \arg\min_\theta
      \enVert{\hat w^{(k)}-\mu_w(\theta)}^{2}_{{K_w(\theta)}^{-1}}+ \trace\cbr{{K_w(\theta)}^{-1} \hat P^{(k)}_w} + \log \det K_w(\theta),
    \end{aligned}
\end{equation}
and the noise variances can be updated as
\begin{equation}
  \begin{aligned}
    \hat \sigma_v^{(k+1)} &= \frac{\hat R^{(k)}_v}{N},\\
    \hat \sigma_y^{(k+1)} &=
  \frac{\hat R_y^{(k)}\!+\!
    \enVert[1]{\hat g^{(k)}}^2_{\hat S_w^{(k)}} \!+\! \trace\cbr{\hat T_w^{(k)}
  P_g^{(k)}}}{N}.
  \end{aligned}
\end{equation}
\end{cor}
\begin{proof}
  Follows from direct maximization of~\eqref{eq:Qvb}.
\end{proof}

Thanks to Theorem~\ref{thm:Qvb} and Corollary~\ref{cor:Qvb}, we have a simple
iterative proceduce to compute the VBEM estimates of the kernel hyperparameters
and of the noise variances; starting from an inital value of the
hyperparameters, we iterate the following two steps:
\begin{enumerate}
  \item Compute the moments of the variational approximation according to
    Corollary~\ref{cor:iterative_vb}.
  \item Update the hyperparameters according to Corollary~\ref{cor:Qvb}.
\end{enumerate}

Under mild regularity conditions, these iterations yield a sequence of
parameter estimates that converges to a stationary point of the marginal
likelihood of the data (see~\cite[Section~2.2]{beal2003variational}).
Then, we can run the iterations in Corollary~\ref{cor:iterative_vb} again to find the
posterior mean estimates of $g$ and $w$.

\section{Simulations}
In this section, we evaluate the methods proposed on some problems that can be
cast as problems of identifying uncertain-input systems.

\subsection{Cascaded linear systems}
In this numerical experiment, we estimate cascaded systems with the structure
presented in Section~\ref{sec:cascaded_systems}. We perform a Monte Carlo
experiment consisting of 500 runs. In each run, we generate two systems by
randomly sampling 40 poles and 40 zeros, in complex conjugate pairs, using the
following technique. We sample the poles randomly, with magnitudes uniformly between 0.4
and 0.8 and phases uniformly between 0 and $\pi$. We sample the zeros randomly,
with magnitudes uniformly between 0 and 0.92 and phases uniformly between 0 and
$\pi$. All systems are generated with unitary static gain. The noise variances on the
input and output measurements are $1$, respectively $1/100$, times the variance
of the corresponding noiseless signals; this means that the
sensor at the output of $S_2$ is considerably more accurate than the sensor at
the output of $S_1$.

We simulate the responses of the systems with a Gaussian
white-noise input with variance 1. We collect $N=200$ samples of the output, from zero
initial conditions, and we estimate the samples of the impulse responses of the
two systems.

As described in Section~\ref{sec:cascaded_systems}, the systems are modeled as
zero-mean Gaussian processes with first order stable-spline kernels.
All the methods are initialized with the choices $\rho_1=\theta_1=1$ and
$\rho_2=\theta_2 = 0.6$. The noise variances are initialized from the
sample variances of the errors of the linear least squares estimates of $g_1$
and $g_2$ from the noisy data.

In the experiment, we compare the following estimators.
\begin{description}
  \item[C-MCEM] The method described in Section~\ref{sec:mcmc}. It uses an
    MCMC approximation of the joint posterior with $B=400$ and $M=2000$. The
    EM iterations are stopped once the relative change in the parameter values
  is below $10^{-2}$.
  \item[C-VBEM] The method described in Section~\ref{sec:vb}. It uses a
    variational approximation of the joint posterior. The
    EM iterations are stopped once the relative change in the parameter values
  is below $10^{-2}$.
\item[C-2Stage] A kernel-based two-stage method. First, it estimates the
  first system in the cascade from $u$ and $v$. Then, it simulates the
  intermediate signal $\hat w$ as the response of the estimated system to  $u$
  and uses $\hat w$ and $y$ to estimate the second system in the cascade.
\item[C-Naive] A naive kernel-based estimation method. It estimates the
  first system in the cascade from $u$ and $v_t$ and the second system from
  $v_t$ and $y_t$. It corresponds to using the noisy signal $v_t$ as if it were
  the noiseless input to the second system in the cascade.
\end{description}

To evaluate the performance of the estimators, we use the following goodness-of-fit metric
\begin{equation}\label{eq:goodnessoffit}
  \mathrm{Fit}^g_j = 1- \frac{\enVert{g_j - \hat g_j}}{\enVert{g_j -
  \mathrm{mean}(g_j)}}
\end{equation}
where $g_j$ is the impulse response of the system at the $j$th Monte Carlo run,
and  $\hat g_j$ is an estimate of the same impulse response.

The results of the experiment are presented in
Figure~\ref{fig:boxplot_cascade}. The figure shows the boxplots of the fit of
the estimated impulse responses of the two blocks in the cascade over the
systems in the dataset.

\begin{figure}[htb]
  \centering
\includegraphics{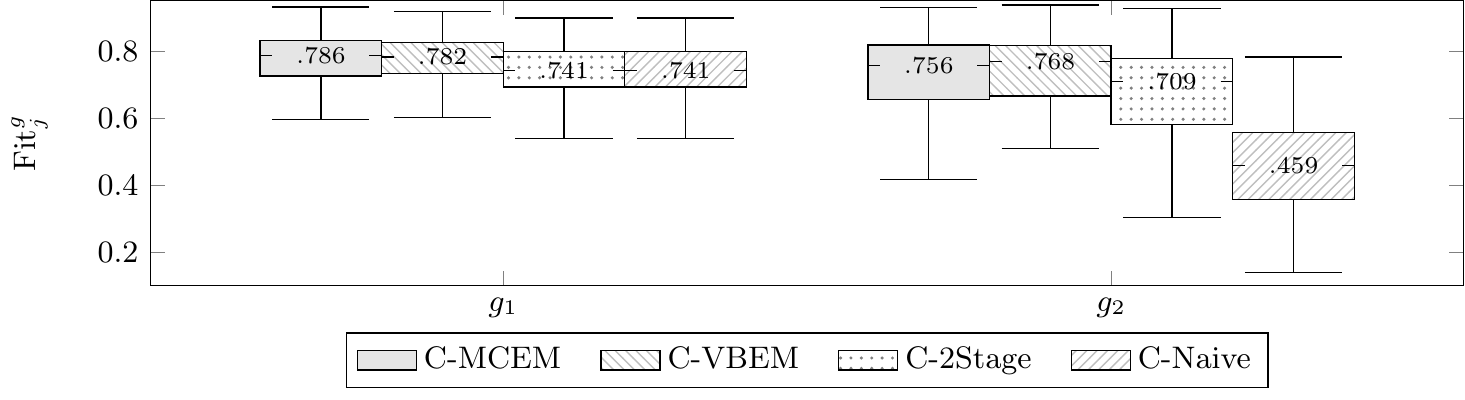}
  \caption{Results of the estimation of cascaded linear systems.}\label{fig:boxplot_cascade}
\end{figure}

From the figure, it appears that the proposed approximation methods are able to
reconstruct the cascaded model with higher accuracy than the alternative
approaches we have considered. Furthermore, there seems to be no clear
disadvantage in using the variational Bayes approximation as compared to the,
more correct, sampling-based approximation. Regarding the performance of the
methods in estimating $g_1$, we see that the methods C-MCEM and C-VBEM perform
better than the other methods (which give the same result). Both C-2Stage and
C-Naive only use the information in $v$ to estimate $g_1$, whereas C-MCEM and
C-VBEM use the full joint distribution of $v$ and $y$ to estimate $g_1$. Given
that in our setting the noise on $y$ is much lower than the noise on $v$,
there is information in $y$ that the joint methods are able to
leverage to improve the estimate of $g_1$ (similar phenomena were already
observed in~\cite{hjalmarsson2009system}, and in~\cite{everitt2013geometric}).
This allows C-MCEM and C-VBEM to better estimate $g_1$.

\subsection{Hammerstein systems}
In this numerical experiment, we estimate Hammerstein systems with the
structure presented in Section~\ref{sec:hammersteins}. We perform four Monte
Carlo experiments consisting of 500 runs. In each run, we generate a stable
transfer-function model by sampling poles and zeros in the complex plane. We
sample the poles, uniformly in magnitude and phase, in the annulus of radii 0.4
and 0.8. We sample the zeros uniformly in the disk of radius 0.92. We generate
the nonlinear transformation as a finite combination of Legendre polynomials
defined as
\begin{equation}
  \varphi_j(x) = 2^j\cdot \sum_{k=0}^j
  x^k\binom{j}{x}\binom{\tfrac{j+k-1}{2}}{j}.
\end{equation}
We sample the coefficients of the combination independently and uniformly in
the interval $\intcc{-1,1}$.

In each Monte Carlo experiment, we consider Hammerstein systems with
different orders for both the nonlinear system and the polynomial nonlinearity.
In Table~\ref{tab:hammerstein_orders}, we present the orders of the systems
considered in the various experiments.

\begin{table}[ht]
  \centering
  \caption{Orders of the Hammerstein systems used in the
  simulations.}\label{tab:hammerstein_orders}
  \begin{tabular}{ccc}
    \toprule
    Dataset & $S$ & $f(\cdot)$ \\
    \midrule
    LOLO (Low-Low) & $\{3,\ldots, 5\}$ & $\{5,\ldots, 10\}$\\
    HILO (High-Low) & $\{9,\ldots, 20\}$ & $\{5,\ldots, 10\}$\\
    LOHI (Low-High) & $\{3,\ldots, 5\}$ & $\{15,\ldots, 20\}$\\
    HIHI (High-High) & $\{9,\ldots, 20\}$ & $\{15,\ldots, 20\}$\\
    \bottomrule
  \end{tabular}
\end{table}

We simulate the responses of the systems in the datasets to a uniform white
noise input in the interval $\sbr{-1,1}$. We collect $N=200$ samples of the
output, from zero initial conditions, and we estimate the static nonlinearity
and the impulse response.

As described in Section~\ref{sec:hammersteins}, the linear blocks are modeled
as zero-mean Gaussian processes with first order stable-spline kernels. We
consider both a parametric model and a nonparametric model for the static
nonlinearity.  All the methods are initialized with $\rho_1 = 1$, $\rho_2=
0.6$. The noise variances are initialized from the prediction error of an
overparameterized least-squares estimate
(see~\cite{bai1998optimal,risuleo2015new}).

In the simulation, we compare the performance of the following estimators:

\begin{description}
  \item[H-P] A semiparametric model for the Hammerstein
    system. It uses the Legendre polynomial basis to construct a linear
    parameterization (with the correct order) of the input:
    \begin{equation}
      \mu_w(\theta)=\Phi\theta,\qquad \sbr[1]{\Phi}_{i,j} = \varphi_j(u_i).
    \end{equation}
    The dynamical system is modeled as a zero-mean Gaussian process with
    covariance matrix given by the first order stable-spline kernel.
  \item[H-MCEM] A nonparametric model for the Hammerstein system with Gibbs
    sampling from the joint posterior with $B=200$ and $M=500$. It uses the
    radial-basis-function kernel~\eqref{eq:rbf} to model the input
    nonlinearity. Note that, because the Hammerstein system is not
    identifiable, we fix $\theta_1=1$ in the algorithm.
  \item[H-VBEM] A nonparametric model for the Hammerstein system with
    variational-Bayes approximation of the joint posterior. It uses the same
    kernel as H-MCEM to model the input  nonlinearity.
  \item[NLHW] The parametric model in Matlab with the default parameters. It
    corresponds to the maximum-likelihood estimator of the model with the
    correct parameterization.
\end{description}
In all methods, the EM iterations are stopped once the relative change in the
parameter values is below $10^{-2}$.

To evaluate the performance of the methods, we use the standard goodness-of-fit
criterion~\eqref{eq:goodnessoffit} for the impulse response of the linear
system. For the input nonlinearity, we compute the
estimated value $\hat w_j$ on a uniform grid of 300 values between -1 and 1 and
we compare it to the true value $w_j$ according to
\begin{equation}
  \mathrm{Fit}^f_j = 1 - \frac{\enVert{w_j - \hat w_j}}{\enVert{w_j -
  \mathrm{mean}(w_j)}},
\end{equation}
where $w_j$ is the vector of values of the true nonlinearity ot the $j$th Monte
Carlo run, and $\hat w_j$ is an estimate of the same vector of values.

\begin{figure*}[htb]
  \centering
\includegraphics{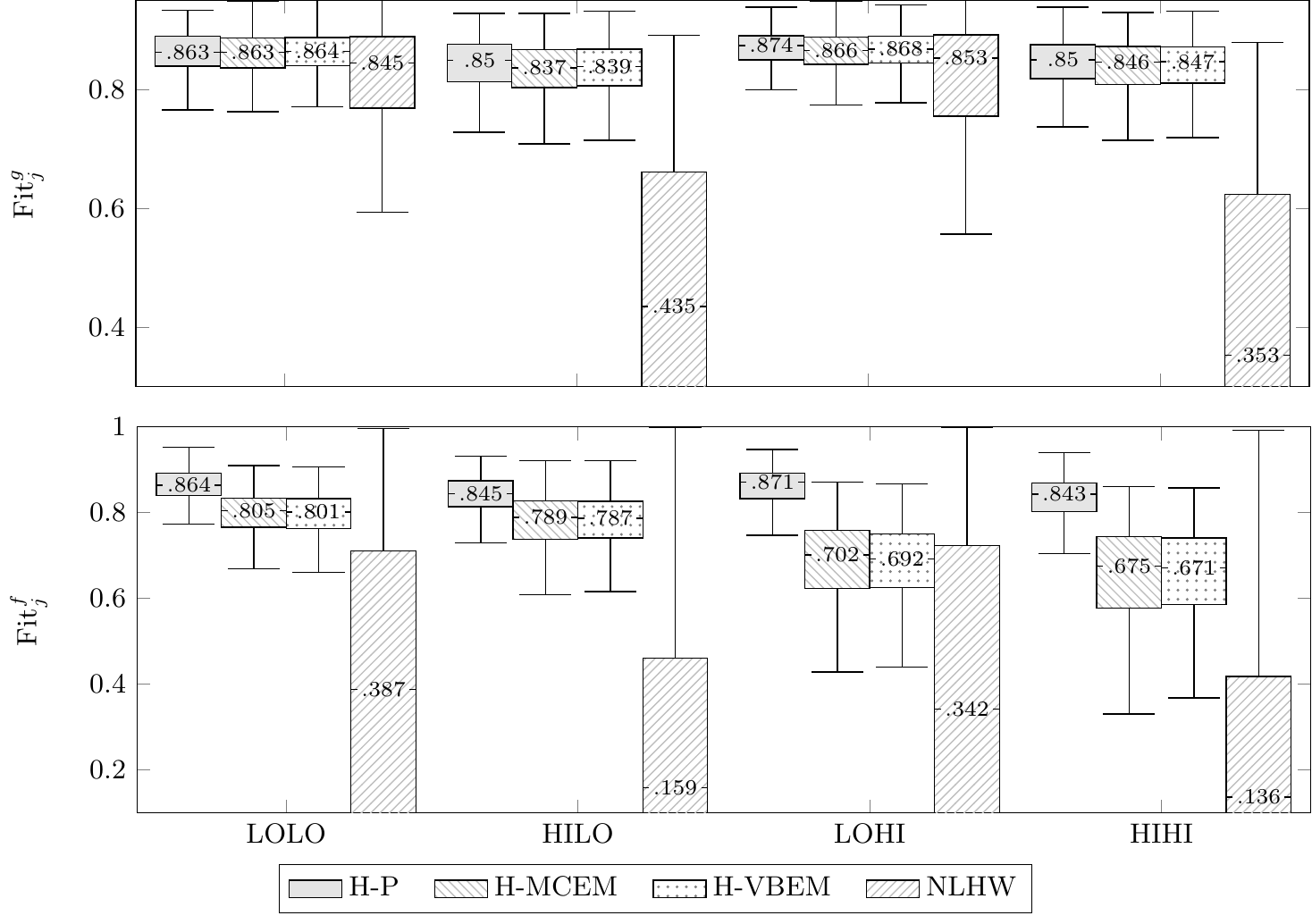}
  \caption{Boxplot of the estimation result}\label{fig:boxplot_hammerstein}
\end{figure*}

The result of the experiment are presented in
Figure~\ref{fig:boxplot_hammerstein}. The figure shows the boxplots of the fit
of the estimated impulse responses (upper pane) and of the static
nonlinearities (lower pane) over the systems in the datasets.

From this simulation, it appears that the proposed nonparametric models are
capable of recovering the system better than the fully parametric NLHW\@. In
addition, it appears that using the correct parametric model for the
input nonlinearity is beneficial in terms of accuracy. As was the case in the
cascaded-system estimation problem, the two approximation methods have comparable performance.

\section{Conclusions}
In this work, we have proposed a new model structure, which we have called the
\emph{uncertain-input model}. Uncertain-input models describe
linear systems subject to inputs about which we have limited information. To
encode the information we have available about the input and the system, we
have used Gaussian-process models.

We have shown how classical problems in system identification can be seen as
uncertain-input estimation problems. Among these applications we find classical
PEM, errors-in-variables and blind system-identification problems,
identification of cascaded linear systems, and identification of Hammerstein
models.

We have proposed an iterative algorithm to estimate the uncertain-input model.
We estimate the impulse response of the linear system and the input
nonlinearity as the posterior means of the Gaussian-process models given the
data. The hyperparameters of the Gaussian-process models are estimated using
the marginal-likelihood method. To solve the related optimization
problem, we have proposed an iterative method based on the EM method.

In the general formulation, the model depends on the convolution of two
Gaussian processes. Therefore, the joint distribution of the data is not
available in closed form.  To circumvent this issue, we have proposed specialized
models, namely the semiparametric and the parametric models, for which the
integrals defining the posterior distributions are available. In the more
general case, we have proposed two approximation methods for the joint
posterior distribution. In the first method, we have used a particle
approximation of the posterior distribution. The particles are drawn using the
Gibbs sampler from Gaussian full-conditional distributions. In the second
method, we have used the variational-Bayes approach to approximate the
posterior distribution. Using a mean-field approximation, we have found that
the posterior distribution can be approximated as a product of two independent
Gaussian random variables.

We have tested the proposed model on two problems: the estimation of cascaded
linear systems and of Hammerstein models. In both cases, the proposed
uncertain-input formulation is able to capture the systems and to provide good
estimates.

Although hinged on the EM method (which is guaranteed to converge under certain
smoothness assumptions) the approximate methods we have proposed do
not have general convergence guarantees: in the formulation given
by~\eqref{eq:ui_model}, there may instances of uncertain-input models for which
the assumptions required for convergence may not hold.  In future publications,
we plan to analyze whether there exists general conditions on the
uncertain-input model such that the algorithms are guaranteed to converge to
optimal solutions.

In addition, the uncertain-input model can be nonidentifiable in certain
configurations (for instance, consider the general errors-in-variables
problem). We plan to further explore this nonidentifiability. Connections with
other problems sharing the same bilinear
structure~\cite{bai2005least,wang2009revisiting} outside of the system
identification framework are also under investigation.

\appendix

\section{Proofs of the main results}\label{app:full_Q}

\subsection{Proof of Theorem~\ref{thm:Qpigs}}\label{pf:Qpigs}
We consider the complete-data likelihood $\p(y,w,g;\tau)$ where $g$ acts as
latent variables. We have that
\begin{equation}
  \begin{aligned}
    \p(y,v,g;\tau) &= \int \p(y,v,g,w;\tau)\dif w= \int
    \p(y|g,w;\sigma_y^2)\p(v|w;\sigma_v^2)\p(g;\rho)\p(w;\theta)\dif w \\
                   &=
    \p(y|g,w=\mu_w(\theta);\sigma_y^2)\p(v|w=\mu_w(\theta);\sigma_v^2)\p(g;\rho),
  \end{aligned}
\end{equation}
where we have used the sampling property of the Dirac density. Hence,
\begin{equation}\label{eq:pfQpigs1}
\begin{aligned}
  \log \p(y,v,g;\tau) &= -\frac{1}{2\sigma_y^2}\enVert{y - G\mu_w}^2
  - \frac{N}{2} \log \sigma_y^2- \frac{1}{2\sigma_v^2} \enVert{v -
  \mu_w}^2 - \frac{N}{2} \log
  \sigma_v^2\\
  & - \frac{1}{2} \enVert{g - \mu_g}_{{K_g}^{-1}}^2 -
  \frac{1}{2} \log \det K_g\,,
\end{aligned}
\end{equation}
where we have dropped explicit dependencies on the hyperparameters.
Taking expectations with respect to $\p(g|y,v;\hat \theta^{(k)})$, we have that
$\E\cbr[1]{\enVert{y- G \mu_w}^2\!} \!=\! y^T y - 2y^T \E\cbr{G}
  \mu_w + \mu_w^T \E\cbr[1]{G^T G} \mu_w$.
The matrix $R$ in~\eqref{eq:R} is such that $G^T = R(I_N\otimes
g)$; hence, we have that
\begin{equation}
  \begin{aligned}
    \E\{G^T\} &= \E\{R(I_N\otimes g)\} = R(I_N\otimes  \hat g^{(k)}) = 
    \hat G^{(k)T},\\
    \E\cbr[1]{G^T G} &= \E\cbr{R(I_N\otimes g)(I_N \otimes g^T) R} = R(I_N\otimes \E\{gg^T\})R^T
    = R(I_N\otimes \hat P_g^{(k)} + \hat g^{(k)}\hat g^{(k)
T})R^T\\
            & = R(I_N\otimes \hat P_g^{(k)})R + R(I_N\otimes \hat g^{(k)}
\hat g^{(k) T})R^T= \hat S_g^{(k)} + R(I_N\otimes \hat g^{(k)})
(I_N\otimes \hat g^{(k) T})R^T\\
&= \hat S_g^{(k)} + \hat G^{(k)T}\hat G^{(k)};
  \end{aligned}
\end{equation}
hence, $\E\cbr[1]{\enVert{y- G \mu_w}^2\!} \!=\! \enVert[1]{y - \hat G^{(k)} \mu_w}^2 +
  \enVert{\mu_w}^2_{\hat S_g^{(k)}}$.

Similarly,
\begin{equation}
  \begin{aligned}
    \E&\cbr{\enVert{g-\mu_g}^2\!}\!\!=\!  
    \trace\cbr[2]{K_g^{-1}\!\!\sbr{\E\{g g^T\}\!-\!2\mu_g \E\{g^T\} \!+\! \mu_g
    \mu_g^T}\!}\\
    &= \trace\cbr[2]{K_g^{-1}\sbr{\hat P_g^{(k)} + \hat g ^{(k)}g^{(k)T}-2\mu_g
\hat g^{(k)T} + \mu_g \mu_g^T}}\\
    &= \enVert{\hat g^{(k)} - \mu_g}^2_{K_{g}^{-1}} +
\trace\cbr[2]{K_g^{-1}\hat P_g^{(k)}}.
    \end{aligned}
\end{equation}
Plugging these expressions into the expectation of~\eqref{eq:pfQpigs1} we
find~\eqref{eq:Qpigs}.

\subsection{Proof of Theorem~\ref{thm:Qmc}}\label{pf:Qmc}
Let $\bar g^{(j)}$ and $\bar w^{(j)}$ be samples draw  from the stationary distribution of the
Gibbs sampler with hyperparameters $\hat \tau$. Now, consider the complete-data likelihood
\begin{equation}\label{eq:completelikelihood_appendix}
  \mathmakebox[0.864\columnwidth][r]{%
\begin{aligned}
  & \log \p(y,v|w,g;\tau)\p(g;\rho)\p(w;\theta) =- \frac{1}{2\sigma_v^2} \enVert{v - w}^2
  - \frac{N}{2} \log
  \sigma_v^2
  -\frac{1}{2\sigma_y^2}\enVert{y - Wg}^2 - \frac{N}{2} \log \sigma_y^2\\
  & \quad - \frac{1}{2} \enVert{g - \mu_g}_{K_g}^2 -
  \frac{1}{2} \log \det K_g - \frac{1}{2} \enVert{w - \mu_w}_{K_w}^2 -
  \frac{1}{2} \log \det K_w\,,
\end{aligned}}
\end{equation}
where we have dropped the explicit dependencies on the hyperparameters. We have that
\begin{equation}\label{eq:pfQmc1}
  \mathmakebox[0.8\columnwidth]{
\begin{aligned}
  &Q^\mc (\tau,\hat\tau) := -\frac{1}{2M\sigma_v^2}  \sum_{j=1}^{M}\enVert{v-
\bar w^{(j)}}^2-
  \frac{N}{2} \log \sigma_v^2 - \frac{1}{2M\sigma_y^2} \sum_{j=1}^M
  \enVert{y - \bar W^{(j)} \bar g^{(j)}}^2 - \frac{N}{2} \log \sigma_y^2\\
  & \quad - \frac{1}{2M}\sum_{j=1}^M \enVert{\bar g^{(j)} - \mu_g}_{{K_g}^{-1}}^2 -
  \frac{1}{2} \log \det K_g - \frac{1}{2M}\sum_{j=1}^M \enVert{\bar w^{(j)} - \mu_w}_{{K_w}^{-1}}^2 -
  \frac{1}{2} \log \det K_w\,.
\end{aligned}}
\end{equation}
Using the definitions in~\eqref{eq:mcem_moments}, we have that
\begin{equation}
  \begin{aligned}
    \sum_{j=1}^M&\enVert[1]{\bar g^{(j)} \!-\! \mu_g}^2_{{K_g}^{-1}} \!=\!
  \sum_{j=1}^M\trace\cbr{K_g^{-1}(\bar g^{(j)} \!-\! \mu_g){(\bar g^{(j)}
      \!\!-\!  \mu_g)}^T}\\
      &= \sum_{j=1}^M\trace\cbr{K_g^{-1}(\bar g^{(j)} \!-\!\hat g \!+\! \hat g \!-\! \mu_g){(\bar g^{(j)}
       \!-\! \hat g \!+\!  \hat g \!-\!\mu_g)}^T}\\
       &= \sum_{j=1}^M\trace\cbr{K_g^{-1}(\bar g^{(j)} \!-\!\hat g)(\bar g^{(j)}
     \!-\! \hat g) }  +M\enVert{\hat g \!-\! \mu_g}_{K_g^{-1}}^2\\
       &= \trace\cbr{K_g^{-1}\sum_{j=1}^M(\bar g^{(j)} \!-\!\hat g)(\bar g^{(j)}
     \!-\! \hat g) }  +M\enVert{\hat g \!-\! \mu_g}_{K_g^{-1}}^2\\
       &= M\trace\cbr{K_g^{-1}\hat P_g}  +M\enVert{\hat g \!-\!
   \mu_g}_{K_g^{-1}}^2;
  \end{aligned}
\end{equation}
similarly, 
\begin{equation}
    \sum_{j=1}^M\enVert[1]{\bar w^{(j)} \!-\! \mu_w}^2_{{K_w}^{-1}} \!=\!
       M\trace\cbr{K_w^{-2}\hat P_w}  +M\enVert{\hat w \!-\!
   \mu_w}_{K_w^{-1}}^2;
\end{equation}
Plugging these expressions into~\eqref{eq:pfQmc1} we
obtain~\eqref{eq:mcem_moments}.

\subsection{Proof of Theorem~\ref{thm:vb_gaussian}}\label{pf:vb_gaussian}
Consider the complete-data likelihood~\eqref{eq:completelikelihood_appendix}.
From~\eqref{eq:q_w_star} we have that
$\log \q_w^\star = \E\cbr{\log p(y,v,g,w;\tau)}$,
where the expectation is taken with respect to $\q_g^\star$. Then, disregarding
terms independent of $w$, we have that
\begin{equation}
    \log \q_w^\star \cong \E\cbr{-\frac{\enVert{y-Gw}^2}{2\sigma_y^2}
    -\frac{\enVert{v-w}^2}{2\sigma_v^2} - \enVert{w-\mu_w}_{{K_w}^{-1}}^{2} }
  \cong - \frac{1}{2}\enVert{w}_{P_w^{-1}}^2 + w^T m_w,
\end{equation}
where
\begin{equation}\label{eq:pfCor72}
    P_w = \del{\frac{1}{\sigma_y^2}\E\{G^T G\} + \frac{1}{\sigma_v^2} I_n
  +K_w^{-1}}^{-1},\qquad
  m_w = P_w\del{\frac{1}{\sigma_y^2}\E\{G^T\} y + \frac{1}{\sigma_v^2}v +
K_w^{-1} \mu_w}.
\end{equation}
Because it is quadratic in $w$, $\q_w$ is a Gaussian distribution. Similarly,
\begin{equation}
    \log \q_g^\star \cong \E\cbr{-\frac{\enVert{y-Wg}^2}{2\sigma_y^2}
    - \enVert{g-\mu_g}_{{K_g}^{-1}}^{2} }
  \cong - \frac{1}{2}\enVert{g}_{P_g^{-1}}^2 + g^T m_g,
\end{equation}
where
\begin{equation}\label{eq:pfCor71}
    P_g = \del{\frac{1}{\sigma_y^2}\E\{W^T W\} 
  +K_g^{-1}}^{-1},\qquad
  m_g = P_g\del{\frac{1}{\sigma_y^2}\E\{W^T\} y 
 + K_g^{-1} \mu_g}.
\end{equation}
and where all expectations are taken with respect to $\q^\star_w$. Because it
is quadratic in $g$, $\q_g$ is also a Gaussian distribution.

\subsection{Proof of Corollary~\ref{cor:iterative_vb}}\label{pf:iterative_vb}
Tracing the proof of Theorem~\ref{thm:vb_gaussian}, we have that $\q_w^{(j+1)}$
is a Gaussian distribution with covariance matrix and mean given
by~\eqref{eq:pfCor72} where the expectations are taken with respect to
$\q_g^{(j)}$. Using the matrix $R$ in~\eqref{eq:R}, we have
\begin{equation} 
  \begin{aligned}
    \E\{G^T G\} &= \E\{R(I_N\otimes g)(I_N \otimes g^T) R^T\}
                = R(I_N\otimes \E\{g g^T\}) R^T = T_g^{(j)},\\
    \E\{G^T\} &= \E\{R(I_N\otimes g)\} =
R(I_N\otimes  g^{(j+1)}) =   G^{(j)T}.
  \end{aligned}
\end{equation}
Similarly, $\q_g^{(j+1)}$ is a Gaussian distribution with covariance matrix and
mean given by~\eqref{eq:pfCor71}, where the expectations are taken with respect
to $\q_w^{(j+1)}$. We have that
\begin{equation} 
  \begin{aligned}
    \E\{W^T W\} &= \E\{R(I_N\otimes w)(I_N \otimes w^T) R^T\}
                = R(I_N\otimes \E\{w w^T\}) R^T = T_w^{(j+1)},\\
    \E\{W^T\} &= 
R(I_N\otimes  w^{(j+1)}) =   W^{(j+1)T}.
  \end{aligned}
\end{equation}
Plugging these expectations into~\eqref{eq:pfCor71} and~\eqref{eq:pfCor72} we
obtain~\eqref{eq:vb_moments}.

\subsection{Proof of Theorem~\ref{thm:Qvb}}\label{pf:Qvb}
We consider again the complete-data
likelihood~\eqref{eq:completelikelihood_appendix}. Taking the expectation with
respect to the independent Gaussian densities $\q_g^{(k)}$ and $\q_w^{(k)}$, we have
that
\begin{align}
  &\E\cbr{\enVert{v - w}^2 }  =  \hat R_v^{(k)} + \trace\,\{
  \hat P^{(k)}\},\\
  &\E\cbr{\enVert{g  -  \mu_g}_{{K_g}^{-1}}^2}  =  \enVert[1]{\hat g^{(k)} -
\mu_g}^2_{{K_g}^{-1}}+ \trace\,\{ K_g^{-1}\hat P_g^{(k)}\},\\
&\E\cbr{\enVert{w  -  \mu_w}_{{K_w}^{-1}}^2}  =  \enVert[1]{\hat w^{(k)} -
\mu_w}^2_{{K_w}^{-1}}+ \trace\,\{ K_w^{-1}\hat P_w^{(k)}\}.
\end{align}
Note that
$\E\{W^T W\} = R^T(I_N\otimes \hat P_w^{(k)} + \hat w^{(k)}\hat w^{(k)T}) R
           = R^T(I_N\otimes \hat P_w^{(k)})R + R^T(I_N\otimes \hat w^{(k)}\hat w^{(k)T}) R
           = \hat S_w^{(k)}  + \hat W^{(k)T}\hat W^{(k)} = \hat T_w^{(k)}$,
and that $\E\{g^T W^T Wg\} = \trace\{\hat T_w^{(k)} (P_g^{(k)} + \hat g^{(k)}
      \hat g^{(k)T})
    =\enVert[1]{\hat W^{(k)}\hat g^{(k)}}^2+ \enVert[1]{\hat g^{(k)}}_{S_w^{(k)}}^2
     + \trace\{\hat T_w^{(k)} \hat P_g^{(k)}\}$,
Hence, $ \E\cbr{\enVert{y  -   W g}^2 }   =  y^T  y  -  2 y^T\E\{Wg\} + \E\{g^T W^T W
  g\}
  = \hat R y^{(k)}
    + \enVert[1]{\hat g^{(k)}}_{S_w^{(k)}}^2 + \trace\,\{\hat T_w^{(k)}
      \hat P_g^{(k)}\}$.
Plugging the terms in the expression of the complete likelihood, we
get~\eqref{eq:Qvb}.

\section{Acknowledgment}
This work was supported by the Swedish Research Council via the projects
\emph{NewLEADS} (contract number: 2016-06079) 
and \emph{System
identification: Unleashing the algorithms} (contract number: 2015-05285), 
and by the European Research Council under the advanced grant
\emph{LEARN} (contract number: 267381).

\bibliographystyle{plain}
\small
\bibliography{uncertain_input}

\begin{thebibliography}{10}

\bibitem{abedmeraim1997blind}
K.~Abed-Meraim, W.~Qiu, and Y.~Hua.
\newblock Blind system identification.
\newblock {\em Proc. IEEE}, 85(8):1310--1322, 1997.

\bibitem{ahmed2014blind}
A.~Ahmed, B.~Recht, and J.~Romberg.
\newblock Blind deconvolution using convex programming.
\newblock {\em IEEE Trans. Inform. Theory}, 60(3):1711--1732, 2014.

\bibitem{bai1998optimal}
E.~W. Bai.
\newblock An optimal two-stage identification algorithm for
  {Hammerstein--Wiener} nonlinear systems.
\newblock {\em Automatica}, 34(3):333--338, 1998.

\bibitem{bai2004convergence}
E.~W. Bai and D.~Li.
\newblock Convergence of the iterative {Hammerstein} system identification
  algorithm.
\newblock {\em IEEE Trans. Autom. Control}, 49(11):1929--1940, 2004.

\bibitem{bai2005least}
E.~W. Bai and Y.~Liu.
\newblock On the least squares solutions of a system of bilinear equations.
\newblock In {\em Proc. {IEEE} Conf. Decis. Control (CDC)}. {IEEE}, 2005.

\bibitem{beal2003variational}
M.~J. Beal.
\newblock {\em Variational Algorithms for Approximate Bayesian Inference}.
\newblock PhD thesis, Gatsby Computational Neuroscience Unit, University
  College London, 2003.

\bibitem{berger1999integrated}
J.~O. Berger, B.~Liseo, and R.~L. Wolpert.
\newblock Integrated likelihood methods for eliminating nuisance parameters.
\newblock {\em Statist. Sci.}, 14(1):1--28, 1999.

\bibitem{bernardo2000bayesian}
J.~M. Bernardo and A.~F.~M. Smith.
\newblock {\em Bayesian Theory}.
\newblock JOHN WILEY \& SONS INC, 2000.

\bibitem{bishop2006pattern}
C.~M. Bishop.
\newblock {\em Pattern Recognition and Machine Learning}.
\newblock Springer, 2006.

\bibitem{bottegal2015blind}
G.~Bottegal, R.~S. Risuleo, and H.~Hjalmarsson.
\newblock Blind system identification using kernel-based methods.
\newblock In {\em Proc. IFAC Symp. System Identification (SYSID)}, volume~48,
  pages 466--471, 2015.

\bibitem{boyd2004convex}
S.~Boyd and L.~Vandenberghe.
\newblock {\em Convex Optimization}.
\newblock Cambridge University Press, 2004.

\bibitem{castaldi1996identification}
P.~Castaldi and U.~Soverini.
\newblock Identification of dynamic errors-in-variables models.
\newblock {\em Automatica}, 32(4):631--636, 1996.

\bibitem{chen2014constructive}
T.~Chen and L.~Ljung.
\newblock Constructive state space model induced kernels for regularized system
  identification.
\newblock In {\em Proc. IFAC World Cong.}, volume~19, pages 1047--1052, 2014.

\bibitem{dempster1977maximum}
A.~P. Dempster, N.~M. Laird, and D.~B. Rubin.
\newblock Maximum likelihood from incomplete data via the em algorithm.
\newblock {\em J. R. Stat. Soc. Ser. B. Stat. Methodol.}, pages 1--38, 1977.

\bibitem{dinuzzo2015kernels}
F.~Dinuzzo.
\newblock Kernels for linear time invariant system identification.
\newblock {\em SIAM J. Control Optim.}, 53(5):3299--3317, 2015.

\bibitem{everitt2013geometric}
N.~Everitt, C.~Rojas, and H.~Hjalmarsson.
\newblock A geometric approach to variance analysis of cascaded systems.
\newblock In {\em Proc. {IEEE} Conf. Decis. Control (CDC)}, 2013.

\bibitem{frigola2014identification}
R.~Frigola, F.~Lindsten, T.~B. Schön, and C.~E. Rasmussen.
\newblock Identification of gaussian process state-space models with particle
  stochastic approximation {EM}.
\newblock {\em {IFAC} Proc. Vol.}, 47(3):4097--4102, 2014.

\bibitem{geman1984stochastic}
S.~Geman and D.~Geman.
\newblock Stochastic relaxation, {Gibbs} distributions, and the {Bayes}ian
  restoration of images.
\newblock {\em IEEE Trans. Pattern Anal. Mach. Intell.}, (6):721--741, 1984.

\bibitem{gilks1996markov}
W.~R. Gilks, S.~Richardson, and D.~J. Spiegelhalter.
\newblock {\em {Markov Chain {Monte Carlo} in Practice}}.
\newblock Chapman and Hall London, 1996.

\bibitem{giri2010block}
F.~Giri and E.~W. Bai.
\newblock {\em Block-oriented nonlinear system identification}.
\newblock Springer, 2010.

\bibitem{hjalmarsson2009system}
H.~Hjalmarsson.
\newblock System identification of complex and structured systems.
\newblock {\em Eur. J. Control}, 15(3-4):275--310, 2009.

\bibitem{kullback1951information}
S.~Kullback and R.~A. Leibler.
\newblock On information and sufficiency.
\newblock {\em Ann. Math. Statist.}, 22(1):79--86, 1951.

\bibitem{linder2017identification}
J.~Linder and M.~Enqvist.
\newblock Identification of systems with unknown inputs using indirect input
  measurements.
\newblock {\em International Journal of Control}, 90(4):729--745, 2017.

\bibitem{ljung1999system}
L.~Ljung.
\newblock {\em System Identification, Theory for the User}.
\newblock Prentice Hall, 1999.

\bibitem{maritz1989empirical}
J.~Maritz and T.~Lwin.
\newblock {\em {Empirical bayes methods}}.
\newblock Chapman and Hall London, 1989.

\bibitem{markovsky2013structured}
I.~Markovsky and K.~Usevich.
\newblock Structured low-rank approximation with missing data.
\newblock {\em SIAM J. Matrix Anal. \& Appl.}, 34(2):814--830, 2013.

\bibitem{mccombie2005laguerre}
D.~B. McCombie, A.~T. Reisner, and H.~H. Asada.
\newblock Laguerre-model blind system identification: Cardiovascular dynamics
  estimated from multiple peripheral circulatory signals.
\newblock {\em IEEE Trans. Biomed. Eng.}, 52(11):1889--1901, 2005.

\bibitem{mclachlan2007algorithm}
G.~McLachlan and T.~Krishnan.
\newblock {\em {The EM algorithm and extensions}}, volume 382.
\newblock John Wiley and Sons, 2007.

\bibitem{meng1993maximum}
X.~L. Meng and D.~B. Rubin.
\newblock {Maximum likelihood estimation via the ECM algorithm: A general
  framework}.
\newblock {\em Biometrika}, 80(2):267--278, 1993.

\bibitem{moulines1995subspace}
E.~Moulines, P.~Duhamel, J.~F. Cardoso, and S.~Mayrargue.
\newblock Subspace methods for the blind identification of multichannel {FIR}
  filters.
\newblock {\em IEEE Trans. Signal Process.}, 43(2):516--525, 1995.

\bibitem{nakajima1993blind}
N.~Nakajima.
\newblock Blind deconvolution using the maximum likelihood estimation and the
  iterative algorithm.
\newblock {\em Opt. Commun.}, 100(1-4):59--66, 1993.

\bibitem{neath2013convergence}
R.~C. Neath.
\newblock On convergence properties of the {Monte Carlo EM} algorithm.
\newblock In {\em Advances in Modern Statistical Theory and Applications: A
  Festschrift in honor of Morris L. Eaton}, pages 43--62. Institute of
  Mathematical Statistics, 2013.

\bibitem{ohlsson2014blind}
H.~Ohlsson, L.~J. Ratliff, R.~Dong, and S.~S. Sastry.
\newblock Blind identification via lifting.
\newblock In {\em Proc. IFAC World Cong.}, 2014.

\bibitem{pillonetto2009bayesian}
G.~Pillonetto and A.~Chiuso.
\newblock A {Bayes}ian learning approach to linear system identification with
  missing data.
\newblock In {\em Proc. IEEE Conf. Decis. Control (CDC)}, pages 4698--4703,
  2009.

\bibitem{pillonetto2014kernel}
G.~Pillonetto, F.~Dinuzzo, T.~Chen, G.~De~Nicolao, and L.~Ljung.
\newblock Kernel methods in system identification, machine learning and
  function estimation: A survey.
\newblock {\em Automatica}, 50(3):657--682, 2014.

\bibitem{pillonetto2011new}
G.~Pillonetto, M.~H. Quang, and A.~Chiuso.
\newblock A new kernel-based approach for nonlinear system identification.
\newblock {\em IEEE Trans. Autom. Control}, 56(12):2825--2840, 2011.

\bibitem{pillonetto2010new}
Gianluigi Pillonetto and Giuseppe De~Nicolao.
\newblock A new kernel-based approach for linear system identification.
\newblock {\em Automatica}, 46(1):81--93, 2010.

\bibitem{rasmussen2006gaussian}
C.~Rasmussen and C.~Williams.
\newblock {\em Gaussian processes for machine learning}.
\newblock the MIT Press, 2006.

\bibitem{risuleo2015kernel}
R.~S. Risuleo, G.~Bottegal, and H.~Hjalmarsson.
\newblock A kernel-based approach to {Hammerstein} system identication.
\newblock In {\em Proc. IFAC Symp. System Identification (SYSID)}, volume~48,
  pages 1011--1016, 2015.

\bibitem{risuleo2015new}
R.~S. Risuleo, G.~Bottegal, and H.~Hjalmarsson.
\newblock A new kernel-based approach to overparameterized {Hammerstein} system
  identification.
\newblock In {\em Proc. {IEEE} Conf. Decis. Control ({CDC})}, pages 115--120,
  2015.

\bibitem{risuleo2015estimation}
R.~S. Risuleo, G.~Bottegal, and H.~Hjalmarsson.
\newblock On the estimation of initial conditions in kernel-based system
  identification.
\newblock In {\em Proc. IEEE Conf. Decis. Control (CDC)}, pages 1120--1125,
  2015.

\bibitem{risuleo2016kernel}
R.~S. Risuleo, G.~Bottegal, and H.~Hjalmarsson.
\newblock Kernel-based system identification from noisy and incomplete
  input-output data.
\newblock In {\em Proc. IEEE Conf. Decis. Control (CDC)}. Institute of
  Electrical and Electronics Engineers ({IEEE}), 2016.

\bibitem{soederstroem2003why}
T.~S{\"o}derstr{\"o}m.
\newblock Why are errors-in-variables problems often tricky?
\newblock In {\em Proc. European Control Conf. (ECC)}, pages 802--807, 2003.

\bibitem{soederstroem2007errors}
T.~S{\"o}derstr{\"o}m.
\newblock Errors-in-variables methods in system identification.
\newblock {\em Automatica}, 43(6):939--958, 2007.

\bibitem{soederstroem2010system}
T.~S{\"o}derstr{\"o}m.
\newblock System identification for the errors-in-variables problem.
\newblock In {\em Proc. UKACC Int. Conf. Control}, pages 1--14, 2010.

\bibitem{svensson2017flexible}
A.~Svensson and T.~B. Sch\"on.
\newblock A flexible state{\textendash}space model for learning nonlinear
  dynamical systems.
\newblock {\em Automatica}, 80:189--199, 2017.

\bibitem{wallin2014maximum}
R.~Wallin and A.~Hansson.
\newblock Maximum likelihood estimation of linear {SISO} models subject to
  missing output data and missing input data.
\newblock {\em Int. J. Control}, pages 1--11, 2014.

\bibitem{wang2009revisiting}
J.~Wang, Q.~Zhang, and L.~Ljung.
\newblock Revisiting the two-stage algorithm for hammerstein system
  identification.
\newblock In {\em Proc. {IEEE} Conf. Decis. Control (CDC)}. {IEEE}, 2009.

\bibitem{wei1990monte}
G.~C.~G. Wei and M.~A. Tanner.
\newblock A {Monte Carlo} implementation of the {EM} algorithm and the poor
  man's data augmentation algorithms.
\newblock {\em Journal of the American Statistical Association},
  85(411):699--704, 1990.

\bibitem{wu1983convergence}
C.~F.~J. Wu.
\newblock On the convergence properties of the {EM} algorithm.
\newblock {\em Ann. Statist.}, 11(1):95--103, 1983.

\bibitem{zhang2015errors}
E.~Zhang and R.~Pintelon.
\newblock Errors-in-variables identification of dynamic systems in general
  cases.
\newblock In {\em Proc. IFAC Symp. System Identification (SYSID)}, volume~48,
  pages 309--313, 2015.

\end{thebibliography}

\end{document}